\documentclass[a4paper]{article}
\usepackage{a4wide}
\usepackage{graphicx,xcolor}
\usepackage{caption}
\usepackage{subcaption}
\usepackage{empheq}
\usepackage{amsthm}
\usepackage{amssymb, latexsym, mathrsfs}
\usepackage{amsmath}
\usepackage{dsfont}
\usepackage{epstopdf}
\usepackage{enumerate}
\usepackage{algorithm}
\usepackage{color}
\usepackage{transparent}
\usepackage{subfig}
\usepackage{url}
\usepackage[affil-it]{authblk}
\usepackage{booktabs}

\usepackage[american,cuteinductors,smartlabels]{circuitikz}
\usepackage{tikz}
\usepackage{schemabloc}
\usetikzlibrary{shapes,arrows}

\theoremstyle{plain}

\newtheorem{rmk}{Remark}
\newtheorem{prop}{Proposition}

\newtheorem{defn}{Definition}

\def\Rset{\mathbb{R}}

\def\Lset{\mathbb{L}}

\newcommand{\AAA}{{\mathcal A}}

\newcommand{\CC}{{\mathcal C}}

\newcommand{\EE}{{\mathcal E}}

\newcommand{\GG}{{\mathcal G}}

\newcommand{\KK}{{\mathcal K}}

\newcommand{\NN}{{\mathcal N}}

\newcommand{\T}{{\mathcal T}}
\newcommand{\UU}{{\mathcal U}}
\newcommand{\VV}{{\mathcal V}}

\newcommand{\ba}[1]{\begin{array}{#1}}
\newcommand{\ea}{\end{array}}

\begin{document}

\title{\LARGE \bf Kron reduction methods for plug-and-play control of ac islanded microgrids with arbitrary topology}

          \author[1]{Michele Tucci%
       \thanks{Electronic address:
         \texttt{michele.tucci02@universitadipavia.it}; Corresponding author}}
\author[1]{Alessandro Floriduz%
       \thanks{Electronic address: \texttt{alessandro.floriduz012@universitadipavia.it}} }       
 \author[2]{Stefano Riverso%
       \thanks{Electronic address: \texttt{RiversS@utrc.utc.com}}} 
\author[1]{Giancarlo Ferrari-Trecate%
       \thanks{Electronic address: \texttt{giancarlo.ferrari@unipv.it}} }


   \affil[1]{Dipartimento di Ingegneria Industriale e dell'Informazione\\Universit\`a degli Studi di Pavia}
     \affil[2]{United Technologies Research Center Ireland}  
     
\date{\textbf{Technical Report}\\ October, 2015}

     \maketitle

     \begin{abstract}
In this paper, we provide an extension of the scalable algorithm proposed in \cite{riverso2015plug} for the design of Plug-and-Play (PnP) controllers for AC Islanded microGrids (ImGs). The method in \cite{riverso2015plug}
assumes Distributed Generation Units (DGUs) are arranged in a load-connected topology, i.e. loads can appear only at the output terminals of inverters. For handling totally general interconnections of DGUs and loads, we describe an approach based on Kron Reduction (KR), a network reduction method giving an equivalent load-connected model of the original ImG. However, existing KR approaches can fail in preserving the structure of transfer functions representing transmission lines. To avoid this drawback, we introduce an approximate KR algorithm, still capable to represent exactly the asymptotic periodic behavior of electric signals even if they are unbalanced. 
Our results are backed up with simulations illustrating features of the new KR approach as well as its use for designing PnP controllers in a 21-bus ImG derived from an IEEE test feeder.

     \end{abstract}
     
     \newpage
     \section{Introduction}
          \label{sec:intro}
   
Kron Reduction (KR) is one of the most popular methods for simplifying linear electrical networks \cite{kron1939tensor} while preserving the behavior of electrical variables at target nodes. Widely used with phasor voltages and currents, KR assumes network nodes are classified either as internal or boundary nodes and it provides an algebraic procedure for computing: 
(i) the topology of a new network connecting boundary nodes only, (ii) the value of admittances related to new edges and (iii) equivalent currents supplied at boundary nodes  accounting for the effect of internal currents in the original network.    
Graph-theoretical properties of KR have been analyzed in \cite{van2010characterization} and \cite{dorfler2013kron} for DC resistive networks. A generalization to AC three-phase balanced circuits in Periodic Sinusoidal Steady State (PSSS), which will be termed AC-KR, can be found in \cite{caliskan2012kron}.
Recently, several studies focused on generalizations of KR preserving the value of electric boundary variables not only in the asymptotic regime, but also during transients \cite{van2010characterization, caliskan2012kron, caliskan2014towards, luo2014spatiotemporal, dhople2014synchronization}. In these \textit{instantaneous} KR procedures, line admittances are replaced by differential models and sufficient conditions guaranteeing well-posedness of the network reduction process have been studied.  

KR finds applications also to Islanded microGrids (ImGs), i.e. autonomous energy islands disconnected from the main grid and composed of the interconnection of DGUs and loads \cite{guerrero2013advanced}. For instance, in \cite{luo2014spatiotemporal} and \cite{schiffer2014conditions} KR is advocated as the key tool for mapping ImGs with a general topology into a standard topology (called \textit{load-connected}) where loads appear only at the terminals of inverters. This can be done labeling the nodes of loads not directly connected to inverter terminals as internal nodes and applying KR to eliminate them. As a consequence, KR can also impact on the design of decentralized control schemes where local controllers associated to DGUs must guarantee voltage and frequency stability of the whole ImG. Indeed, design methods for load-connected ImGs could be directly extended to any ImG by performing control design on the reduced network. However, decentralized control schemes for ImGs often rely on additional assumptions, such as specific line models. For instance, droop controllers are tailored to either mainly inductive or mainly resistive lines \cite{guerrero2007decentralized}, while Plug-and-Play (PnP) controllers in \cite{riverso2015plug} assume RL lines.
This could be a problem as instantaneous KR does not guarantee reduced line models will have the same structure of the original ones, e.g. simple RL lines could result in reduced lines with a more complex dynamics \cite{dhople2014synchronization}. 
Even though this does not happen in specific line cases 
\cite{van2010characterization,caliskan2012kron}, in general, one faces the problem of devising \textit{approximate} instantaneous KR methods for preserving selected features of line models.

In this paper we focus on adapting KR to the PnP design algorithm in \cite{riverso2015plug} devoted to the synthesis of decentralized controllers for AC ImGs. The main advantage of PnP control is that the computation of a local regulator for a DGU does not require a global ImG model but only parameters of lines connected to that DGU.
Furthermore, one can test if the addition/removal of a DGU can spoil stability of the whole ImG before performing the operation and by solving a local optimization problem. The method in \cite{riverso2015plug} assumes load-connected ImGs and RL lines.
In order to preserve the structure of transfer functions describing lines, we propose an approximate instantaneous KR procedure that we term \emph{hybrid KR (hKR)}, as it combines AC-KR with dynamic line models. We show that hKR reproduces exactly the behavior of voltages and currents at boundary nodes in PSSS, even in the case of unbalanced phases. 
The design of PnP controllers based on hKR is tested on a 21-bus ImG derived from the IEEE test feeder in \cite{feeders2011ieee}, enhanced with switches yielding changes of line topology and plug-in/out of DGUs. Simulations performed in PSCAD confirm the applicability of our approach.

The models of DGUs and lines are introduced in Section \ref{sec:ImGmodel}. Section \ref{sec:MethodsforKR} presents the new hKR approach together with simulations showing its features. Results in \cite{riverso2015plug} on PnP control design are briefly reviewed in Section \ref{sec:PnPctrl}. Simulations illustrating the joint use of hKR and PnP design are given in Section \ref{sec:Sim21bn}.

\textbf{Notation and basic definitions.} We use $f^{abc}(t)=[f_a(t),f_b(t),f_c(t)]^{T}\in \Rset^{3}$ for denoting three-phase signals in the $abc$ frame. 
A three-phase signal $f^{abc}(t)$ is balanced if the analytic signals \cite{bracewell1965fourier} $F_a(t)\exp{\left(\mathrm i \delta_a(t)\right)}$, $F_b(t)\exp{\left(\mathrm i \delta_b(t)\right)}$, $F_c(t)\exp{\left(\mathrm i \delta_c(t)\right)}$ associated  to $f_a(t)$, $f_b(t)$, $f_c(t)$, respectively, have the same envelope (i.e. $F_a(t)=F_b(t)=F_c(t)$) and their instantaneous phases are such that $\delta_b(t)=\delta_a(t)-2\pi/3=\delta_c(t)+2\pi/3$ or $\delta_b(t)=\delta_a(t)+2\pi/3=\delta_c(t)-2\pi/3$; in the first case one has a positive-sequence balanced signal, in the latter a negative-sequence balanced signal \cite{cablea2014method}. A three-phase signal is unbalanced if it is not balanced. A three-phase network element is balanced if it drains balanced currents when supplied by balanced voltages, with currents of the same sequence (positive or negative) of voltages; otherwise, it is said to be unbalanced \cite{schiffer15_mod}.
To $f^{abc}(t)$, we associate its representation in the  $dq0$ reference frame (i.e. $f^{dq0}(t)$). It is obtained from $f^{abc}(t)$ through the Park transformation \cite{park1929two}, denoted with $T(\theta(t))$, $\theta(t) = \omega_{0}t$, $\omega_0$ being the nominal network frequency. Since in the present work only signals without zero-sequence component (i.e. $f^0(t)=0$) are used, we can associate to $f^{abc}(t)$ its complex $dq$-representation $f^{dq}(t)=f^d(t)+\mathrm i \, f^q(t)$ without loss of generality.

When clear from the context, we omit time dependence of electrical quantities. 
$\mathscr L[\cdot]$ is the Laplace-transform operator.
The cardinality of the finite set $S$ will be denoted with $|S|$. According to the definitions in \cite{bollobas1998modern}, a weighted graph $\mathcal G = (\mathcal V, \mathcal E, W)$ is given by a finite set of nodes $\mathcal V=\{1,\dots, n\}$, a set of edges $\mathcal E \subseteq  \mathcal V \times \mathcal V$ and a diagonal matrix $W$ of dimension $|\EE|\times|\EE|$, collecting, on its diagonal, weights $W_{e}$, $e\in\EE$. In this work, weights can be real numbers, complex numbers or SISO transfer functions (in this case, we replace $W$ with $W(s)$). An edge $e\in\EE$ is a self-loop if $e = (i,i)$, for some $i\in\VV$. All graphs in this work will not contain self-loops. The order of a graph is $n=|\VV|$.
A graph is undirected if $(x,y) \in \mathcal E \implies (y,x) \in \mathcal E$. In this case, the pairs $(x,y)$ and $(y,x)$ are considered as identical and unordered. Otherwise, the graph is said to be directed. The set of neighbors of node $i\in\VV$ is $\NN_i=\{j:$ $(i,j)\in\EE$ or $(j,i)\in\EE\}$. A graph is said to be connected if there is a path from any vertex to every other vertex of the graph. All graphs considered in this work are connected.

The incidence matrix of $\GG$ is denoted with $B\in \mathbb{R}^{|\VV|\times|\EE|}$ \cite{bollobas1998modern}. The Laplacian of $\GG$ is the matrix $\Lset = BWB^{T}$. 

     \section{Microgrid model}
     \label{sec:ImGmodel}
\subsection{ImG associated graph}
\label{sec:ImGgraph}
Consider an ImG composed of DGUs and loads, connected through electrical lines. The line network has an arbitrary topology, which we assume to be connected. Without loss of generality, one can represent each ImG as a weighted connected directed graph $\GG = (\VV,\EE,W(s))$ (see Figure \ref{fig:ex_ImG_graphmodel}, for example) where the vertex set $\VV$ is partitioned into a set of boundary nodes $\VV_{b}$ and a set of internal nodes $\VV_{\ell}$. Subset $\VV_b$ identifies Point of Common Coupling (PCC) nodes, i.e. the terminals of each DGU, while $\VV_{\ell}$ contains load nodes, i.e. loads that are not directly connected to PCCs. Similarly, we denote the nodal currents injected by the DGUs as $I_b$; the nodal currents injected by loads are $I_{\ell}$. Nodal currents are positive if entering into nodes. Nodal voltages are partitioned analogously. 

Each edge corresponds to an electric line and the orientation of the edges $ e_1, ..., e_{|\mathcal E|}\in \EE $ is arbitrary. We adopt the following sign convention: reference directions of line currents coincide with edges orientations and voltages $V_{e}$, $e = (i,j)\in\EE$ are defined as $V_i-V_j$. Finally, the weight of every edge is given by its admittance transfer function $W_{e} (s)$ accounting for the dynamics of line $e\in\EE$.
\begin{figure}
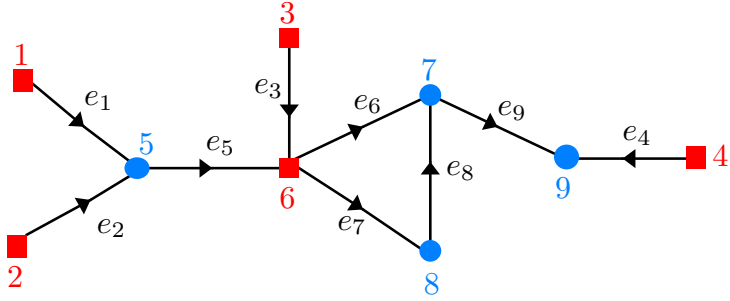

		\centering
		\scalebox{1.2}{\begin{pgfpicture}{0cm}{0cm}{244pt}{118pt}
\pgfsetxvec{\pgfpoint{1pt}{0pt}}
\pgfsetyvec{\pgfpoint{0pt}{1pt}}
\pgfsetroundjoin 
\pgfsetroundcap
\pgftranslateto{\pgfxy(0,118)}
\begin{pgfmagnify}{1}{-1}
\definecolor{layer0}{rgb}{0.0,0.0,0.0}
\definecolor{layer1}{rgb}{0.0,0.0,0.5}
\definecolor{layer2}{rgb}{1.0,0.0,0.0}
\definecolor{layer3}{rgb}{0.0,0.5,0.5}
\definecolor{layer4}{rgb}{1.0,0.78,0.0}
\definecolor{layer5}{rgb}{0.5,1.0,0.0}
\definecolor{layer6}{rgb}{0.0,1.0,1.0}
\definecolor{layer7}{rgb}{0.0,0.5,0.0}
\definecolor{layer8}{rgb}{0.6,0.8,0.2}
\definecolor{layer9}{rgb}{1.0,0.08,0.58}
\definecolor{layer10}{rgb}{0.71,0.61,0.05}
\definecolor{layer11}{rgb}{0.0,0.5,1.0}
\definecolor{layer12}{rgb}{0.88,0.88,0.88}
\definecolor{layer13}{rgb}{0.64,0.64,0.64}
\definecolor{layer14}{rgb}{0.37,0.37,0.37}
\definecolor{layer15}{rgb}{0.0,0.0,0.0}
\color{layer0}
\pgfsetlinewidth{0.815pt}
\pgfline{\pgfxy(11.0,29.0)}{\pgfxy(30.0,45.0)}
\pgfmoveto{\pgfxy(30.0,45.0)}
\pgflineto{\pgfxy(25.007947970183455,44.71819061122003)}
\pgflineto{\pgfxy(28.872762444880134,40.12872342251773)}
\pgfclosepath 
\pgffill 
\pgfline{\pgfxy(30.0,45.0)}{\pgfxy(47.0,58.0)}
\pgfline{\pgfxy(47.0,58.0)}{\pgfxy(70.0,58.0)}
\pgfmoveto{\pgfxy(70.0,58.0)}
\pgflineto{\pgfxy(66.0,61.0)}
\pgflineto{\pgfxy(66.0,55.0)}
\pgfclosepath 
\pgffill 
\pgfline{\pgfxy(70.0,58.0)}{\pgfxy(92.0,58.0)}
\pgfline{\pgfxy(94.0,17.0)}{\pgfxy(94.0,42.0)}
\pgfmoveto{\pgfxy(94.0,42.0)}
\pgflineto{\pgfxy(91.0,38.0)}
\pgflineto{\pgfxy(97.0,38.0)}
\pgfclosepath 
\pgffill 
\pgfline{\pgfxy(94.0,42.0)}{\pgfxy(94.0,56.0)}
\pgfline{\pgfxy(94.0,56.0)}{\pgfxy(117.0,71.0)}
\pgfmoveto{\pgfxy(117.0,71.0)}
\pgflineto{\pgfxy(112.01075427100477,71.32776066832815)}
\pgflineto{\pgfxy(115.2883609542863,66.30209708729646)}
\pgfclosepath 
\pgffill 
\pgfline{\pgfxy(117.0,71.0)}{\pgfxy(136.0,84.0)}
\pgfline{\pgfxy(12.0,79.0)}{\pgfxy(32.0,68.0)}
\pgfmoveto{\pgfxy(32.0,68.0)}
\pgflineto{\pgfxy(29.94089261460988,72.55632272511856)}
\pgflineto{\pgfxy(27.049380115976945,67.29902727305868)}
\pgfclosepath 
\pgffill 
\pgfline{\pgfxy(32.0,68.0)}{\pgfxy(45.0,60.0)}
\pgfline{\pgfxy(94.0,56.0)}{\pgfxy(117.0,45.0)}
\pgfmoveto{\pgfxy(117.0,45.0)}
\pgflineto{\pgfxy(114.68582960536943,49.43222465412296)}
\pgflineto{\pgfxy(112.0970966215454,44.019419324309084)}
\pgfclosepath 
\pgffill 
\pgfline{\pgfxy(117.0,45.0)}{\pgfxy(138.0,35.0)}
\pgfline{\pgfxy(138.0,35.0)}{\pgfxy(138.0,56.0)}
\pgfline{\pgfxy(138.0,56.0)}{\pgfxy(138.0,84.0)}
\pgfmoveto{\pgfxy(138.0,56.0)}
\pgflineto{\pgfxy(141.0,60.0)}
\pgflineto{\pgfxy(135.0,60.0)}
\pgfclosepath 
\pgffill 
\pgfline{\pgfxy(138.0,35.0)}{\pgfxy(159.0,45.0)}
\pgfmoveto{\pgfxy(159.0,45.0)}
\pgflineto{\pgfxy(154.09875718352725,45.9888472349024)}
\pgflineto{\pgfxy(156.67835866588132,40.57168412195882)}
\pgfclosepath 
\pgffill 
\pgfline{\pgfxy(159.0,45.0)}{\pgfxy(180.0,55.0)}
\pgfline{\pgfxy(219.0,55.0)}{\pgfxy(198.0,55.0)}
\pgfmoveto{\pgfxy(198.0,55.0)}
\pgflineto{\pgfxy(202.0,52.0)}
\pgflineto{\pgfxy(202.0,58.0)}
\pgfclosepath 
\pgffill 
\pgfline{\pgfxy(198.0,55.0)}{\pgfxy(180.0,55.0)}
\begin{pgfmagnify}{1}{-1}
\pgfputat{\pgfxy(30,-34)}{\pgfbox[left,top]{$e_1$}}
\end{pgfmagnify}
\begin{pgfmagnify}{1}{-1}
\pgfputat{\pgfxy(34,-74)}{\pgfbox[left,top]{$e_2$}}
\end{pgfmagnify}
\begin{pgfmagnify}{1}{-1}
\pgfputat{\pgfxy(83,-30)}{\pgfbox[left,top]{$e_3$}}
\end{pgfmagnify}
\begin{pgfmagnify}{1}{-1}
\pgfputat{\pgfxy(198,-45)}{\pgfbox[left,top]{$e_4$}}
\end{pgfmagnify}
\begin{pgfmagnify}{1}{-1}
\pgfputat{\pgfxy(68,-48)}{\pgfbox[left,top]{$e_5$}}
\end{pgfmagnify}
\begin{pgfmagnify}{1}{-1}
\pgfputat{\pgfxy(114,-35)}{\pgfbox[left,top]{$e_6$}}
\end{pgfmagnify}
\begin{pgfmagnify}{1}{-1}
\pgfputat{\pgfxy(109,-73)}{\pgfbox[left,top]{$e_7$}}
\end{pgfmagnify}
\begin{pgfmagnify}{1}{-1}
\pgfputat{\pgfxy(143,-56)}{\pgfbox[left,top]{$e_8$}}
\end{pgfmagnify}
\begin{pgfmagnify}{1}{-1}
\pgfputat{\pgfxy(159,-37)}{\pgfbox[left,top]{$e_9$}}
\end{pgfmagnify}
\color{layer2}
\begin{pgfmagnify}{1}{-1}
\pgfputat{\pgfxy(8,-19)}{\pgfbox[left,top]{$1$}}
\end{pgfmagnify}
\begin{pgfmagnify}{1}{-1}
\pgfputat{\pgfxy(6,-89)}{\pgfbox[left,top]{$2$}}
\end{pgfmagnify}
\begin{pgfmagnify}{1}{-1}
\pgfputat{\pgfxy(91,-6)}{\pgfbox[left,top]{$3$}}
\end{pgfmagnify}
\pgfmoveto{\pgfxy(91,14)}
\pgflineto{\pgfxy(97,14)}
\pgflineto{\pgfxy(97,20)}
\pgflineto{\pgfxy(91,20)}
\pgfclosepath 
\pgffill 
\pgfmoveto{\pgfxy(8,27)}
\pgflineto{\pgfxy(14,27)}
\pgflineto{\pgfxy(14,34)}
\pgflineto{\pgfxy(8,34)}
\pgfclosepath 
\pgffill 
\pgfmoveto{\pgfxy(218,51)}
\pgflineto{\pgfxy(224,51)}
\pgflineto{\pgfxy(224,58)}
\pgflineto{\pgfxy(218,58)}
\pgfclosepath 
\pgffill 
\pgfmoveto{\pgfxy(6,79)}
\pgflineto{\pgfxy(12,79)}
\pgflineto{\pgfxy(12,86)}
\pgflineto{\pgfxy(6,86)}
\pgfclosepath 
\pgffill 
\begin{pgfmagnify}{1}{-1}
\pgfputat{\pgfxy(226,-51)}{\pgfbox[left,top]{$4$}}
\end{pgfmagnify}
\begin{pgfmagnify}{1}{-1}
\pgfputat{\pgfxy(91,-64)}{\pgfbox[left,top]{$6$}}
\end{pgfmagnify}
\pgfmoveto{\pgfxy(91,55)}
\pgflineto{\pgfxy(97,55)}
\pgflineto{\pgfxy(97,61)}
\pgflineto{\pgfxy(91,61)}
\pgfclosepath 
\pgffill 
\color{layer11}
\begin{pgfmagnify}{1}{-1}
\pgfputat{\pgfxy(136,-91)}{\pgfbox[left,top]{$8$}}
\end{pgfmagnify}
\begin{pgfmagnify}{1}{-1}
\pgfputat{\pgfxy(177,-61)}{\pgfbox[left,top]{$9$}}
\end{pgfmagnify}
\pgfellipse[fillstroke]{\pgfxy(138.0,35.0)}{\pgfxy(3.0,0)}{\pgfxy(0,3.0)}
\pgfellipse[fillstroke]{\pgfxy(138.0,84.0)}{\pgfxy(3.0,0)}{\pgfxy(0,3.0)}
\pgfellipse[fillstroke]{\pgfxy(180.5,54.5)}{\pgfxy(3.5,0)}{\pgfxy(0,3.5)}
\pgfellipse[fillstroke]{\pgfxy(46.5,58.0)}{\pgfxy(3.5,0)}{\pgfxy(0,3.0)}
\begin{pgfmagnify}{1}{-1}
\pgfputat{\pgfxy(47,-47)}{\pgfbox[left,top]{$5$}}
\end{pgfmagnify}
\begin{pgfmagnify}{1}{-1}
\pgfputat{\pgfxy(135,-24)}{\pgfbox[left,top]{$7$}}
\end{pgfmagnify}
\end{pgfmagnify}
\end{pgfpicture}}
		\caption{Example of a graph representing the line network of a microgrid. Red squares denote boundary nodes (i.e. DGUs with corresponding local loads $I_{Li}$, if any), while blue circles represent internal nodes (i.e. loads).}
		\label{fig:ex_ImG_graphmodel}
              \end{figure}
\subsection{DGU and line electrical models}
     We assume three-phase electrical signals without zero-sequence components and balanced network parameters. 
Note that we do not assume balanced signals, hence including in our framework the case of unbalanced load currents. As in \cite{babazadeh2013robust}, the electrical scheme of DGU $i \in \mathcal V_b$ is represented in Figure \ref{fig:schema1DGU} and its model in \emph{dq} coordinates is:
\begin{equation}
    	\label{eq:DGUeqx}
    	\quad\left\lbrace
    	\begin{aligned}
    	\frac{\mathrm d}{\mathrm dt} V_i^{dq}&=-\mathrm i \omega_0 V_i^{dq}+ \frac{I_{ti}^{dq}}{C_{ti}}-\frac{I_{Li}^{dq}}{C_{ti}}-\frac{1}{C_{ti}}I_{bi}^{dq}\\
    	\frac{\mathrm d}{\mathrm dt}I_{ti}^{dq} &=-\left( \frac{R_{ti}}{L_{ti}}+\mathrm i \omega_0 \right) I_{ti}^{dq} - \frac{V_i^{dq}}{L_{ti}} + \frac{V_{ti}^{dq}}{L_{ti}} 
    	\end{aligned}
    	\right.
    \end{equation}
    where $V_{i}$, $I_{ti}$, $I_{bi}$, $I_{Li}$, $V_{ti}$, $R_{ti}$, $C_{ti}$ and $L_{ti}$ are shown in Figure \ref{fig:schema1DGU}. 
    \begin{rmk}
    	We highlight that each DGU might present a local load current $I_{Li}$ connected to its PCC and we will treat it as an exogenous disturbance in control design. However, local load currents $I_L$ are different from load currents $I_{\ell}$ defined in Section \ref{sec:ImGgraph}. The latter ones, in fact, represent the effect of loads that are not directly connected to the terminals of DGUs. \hspace{70mm} $\blacksquare$
    \end{rmk} 

In \emph{dq} reference frame, the three-phase RL line associated to the edge $(i,j)\in \mathcal E$ has the dynamics:
    \begin{equation*}
    \frac{\mathrm d}{\mathrm dt}I_{ij}^{dq}=-\left( \frac{R_{ij}}{L_{ij}}+\mathrm i \omega_0 \right) I_{ij}^{dq} + \frac{1}{L_{ij}}\left( V_i^{dq}-V_j^{dq}\right) 
    \end{equation*}
    where $I_{ij}$ is shown in Figure \ref{fig:schema1DGU}. Equivalently, in term of transfer functions, one has $I_{ij}^{dq}(s)=W_{ij}(s)\left( V_i^{dq}(s)-V_j^{dq}(s) \right)$ with:
    \begin{equation}
    W_{ij}(s)=\frac{1}{Z_{ij}+L_{ij}s}, \qquad Z_{ij}=R_{ij}+\mathrm i \omega_0 L_{ij}.\quad
    \label{ConstitutiveRel}
    \end{equation}
Consider an ImG composed of $n$ nodes partitioned into $n_b$ boundary nodes and $n_{\ell} = n-n_b$ internal nodes. If $n=n_b$ and $n_{\ell}=0$, then we term the microgrid ``load-connected'' as the current loads appear only at PCC of DGUs. Otherwise, let us set: $I^{dq}=\left[I_{b}^{dq^T}, I_{\ell}^{dq^T}\right]^T$, where $I_b^{dq}=\left[I_{b_1}^{dq}, ..., I_{b_{n_b}}^{dq} \right]^T $ and $I_{\ell}^{dq}=\left[ I_{\ell_1}^{dq}, ..., I_{\ell_{n_{\ell}}}^{dq} \right]^T$. If $b_k =i$, the boundary current $I_{b_k}^{dq}$ is the current $I_{bi}^{dq}$ injected by DGU $i$. Nodal voltages $V^{dq}=\left[ V_b^{dq^T}, V_{\ell}^{dq^T}\right]^T$
are partitioned analogously. In order to account for the network interconnections, by applying KCL and KVL laws one obtains \cite{dhople2014synchronization}:
    \begin{equation}
    \begin{pmatrix}
    I_b^{dq}(s)\\ I_{\ell}^{dq}(s)
    \end{pmatrix}=
    \begin{pmatrix}
    \Lset_{bb}(s) & \Lset_{b\ell}(s)\\ \Lset_{\ell b}(s) & \Lset_{\ell \ell}(s)
    \end{pmatrix}
    \begin{pmatrix}
    V_{b}^{dq}(s)\\ V_{\ell}^{dq}(s)
    \end{pmatrix}
    \label{I=LV}
    \end{equation}
    that is $I^{dq}(s)=\Lset(s) V^{dq}(s)$ where $\Lset(s)$ is the graph Laplacian of the graph $\GG$ with weights $W_{ij}(s)$. By construction, one has $\Lset_{ij}(s)=-W_{ij}(s)$ if $(i,j)\in \mathcal E$ and $\Lset_{ii}(s)=-\sum_{j\in \mathcal N_i}{\Lset_{ij}(s)}$. 
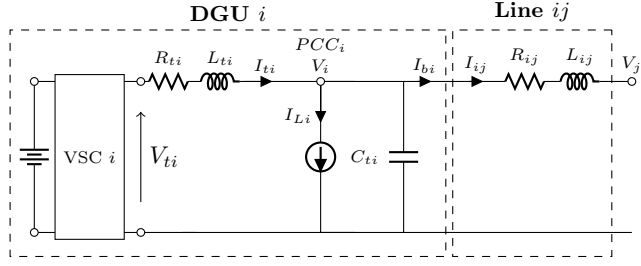
\begin{figure}[htb]
            \centering
            \ctikzset{bipoles/length=0.7cm}
\begin{circuitikz}[scale=.91,transform shape, color=black]
\ctikzset{current/distance=1}


\draw (0,0) node[ocirc](Aibattery) {}
node[ocirc] (Bibattery) at ([xshift=-0cm,yshift=-2.2cm]Aibattery) {}
(Bibattery) to [battery] (Aibattery) {}
node [rectangle,draw,minimum width=1cm,minimum height=2.4cm] (vsci) at ($0.5*(Aibattery)+0.5*(Bibattery)+(0.8,0)$) {\scriptsize{VSC $i$}}
(Aibattery) to [short] ([xshift=0.3cm]Aibattery)
 (Bibattery) to [short] ([xshift=0.3cm]Bibattery)


node[ocirc] (Ai) at ($(Aibattery)+(1.54,0)$) {}
node[ocirc] (Bi) at ($(Bibattery)+(1.54,-0.0)$) {}
(Ai) to [short] ([xshift=-0.24cm]Ai)
(Bi) to [short] ([xshift=-0.24cm]Bi)
(Ai) to [R, l=\scriptsize{$R_{ti}$}] ($(Ai)+(0.8,0)$) {}
to [L, l=\scriptsize{$L_{ti}$}]($(Ai)+(1.5,0)$){}
to [short,i=\scriptsize{$I_{ti}$}]($(Ai)+(1.8,0)$){}
(Bi) to [short] ($(Bi)+(0.8,0)$);
\begin{scope}[shorten >= 10pt,shorten <= 10pt,]
\draw[<-] (Ai) -- node[right] {$V_{ti}$} (Bi);
\end{scope};
 \draw
 ($(Ai)+(2.6,0.3)$) node[anchor=south]{\scriptsize{$PCC_i$}}
 ($(Ai)+(2.6,0)$) node[anchor=south]{{\scriptsize{$V_i$}}}
($(Ai)+(2.6,0)$$) node[ocirc](PCCi)(B1){}

 ($(B1)+(0,0)$) to [I]($(B1)+(0,-2.2)$)--($(B1)+(0.5,-2.2)$)
($(B1)+(0,-0.5)$) to [short,i<=\scriptsize{$I_{Li}$}] (B1)--($(B1)+(0.5,0)$)
($(B1)+(1.2,-2.2)$) to [C, l=\scriptsize{$C_{ti}$}] ($(B1)+(1.2,0)$)

($(B1)+(-0.8,0)$) to [short] ($(B1)+(2.6,0)$)
($(B1)+(-2,-2.2)$) to [short] ($(B1)+(4.5,-2.2)$)
($(B1)+(1.4,0)$) to [short,i^>=\scriptsize{$I_{bi}$}] ($(B1)+(1.5,0)$)
($(B1)+(2.15,0)$) to [short,i^>=\scriptsize{$I_{ij}$}] ($(B1)+(2.25,0)$)
($(B1)+(2.6,0)$) to [R, l=\scriptsize{$R_{ij}$}] ($(B1)+(3.3,0)$) 
($(B1)+(3.3,0)$)  to [L, l=\scriptsize{$L_{ij}$}]($(B1)+(4.2,0)$)
($(B1)+(4.2,0)$) to [short] ($(B1)+(4.5,0)$)


node [rectangle,draw,minimum width=6.3 cm,minimum
height=3.3cm,dashed,label=\small\textbf{DGU $i$}] (DGUi) at
($0.5*(Aibattery)+0.5*(Bibattery)+(2.8,0.2)$) {}
node [rectangle,draw,minimum width=2.3cm,minimum height=3.3cm,dashed,label=\small\textbf{Line $ij$}] (Lineij) at ($0.5*(DGUi.center)+(5.8,-0.45)$){}

($(B1)+(4.5,0)$$) node[ocirc](PCCj)(B2){}
($(B2)$) node[anchor=south]{{\scriptsize{$V_j$}}}

;\end{circuitikz}
            \caption{Equivalent single-phase electrical scheme of DGU $i$ composed of a Voltage Source Converter (VSC), an RLC filter, a local load $I_{Li}$ and a single RL line connecting DGU $i$ with the node $j\in\VV$.}
            \label{fig:schema1DGU}
          \end{figure}

     \section{Methods for KR}
     \label{sec:MethodsforKR}
     KR \cite{kron1939tensor} is a reduction method for eliminating internal nodes in an electrical network while preserving relevant features of voltages and currents at boundary nodes. In our setup, internal nodes correspond to load nodes. 

We first summarize two existing approaches (\emph{instantaneous} KR and AC-KR). Then, we present hKR in Section \ref{subsec:HybridKR} and illustrate its performance through simulations in Section \ref{subsec:ne}.

Assuming $\Lset_{\ell \ell}(s)$ in \eqref{I=LV} is invertible\footnote{Conditions for the invertibility of $\Lset_{\ell\ell}(s)$ have been studied in \cite{dhople2014synchronization}.} for some $s \in \mathbb C$, we have \cite{dhople2014synchronization}:
     \begin{subequations}
     	\label{KRdef}
     	\begin{empheq}{align}
     	\label{KRdef1}
     	I_b^{dq}(s)=\KK\left(\Lset(s) \right) V_b^{dq}(s)- \mathcal T (s)I_{\ell}^{dq}(s)
     	\end{empheq}
     	\begin{empheq}{align}
     	\label{KRdef2}
     	\KK\left(\Lset(s) \right) = \Lset_{bb}(s) -\Lset_{b\ell}(s)\Lset_{\ell \ell}^{-1}(s)\Lset_{\ell b}(s)
     	\end{empheq}
     	\begin{empheq}{align}
     	\label{KRdef3}
     	\T(s)=-\Lset_{b\ell}(s)\Lset_{\ell \ell}^{-1}(s).
     	\end{empheq}
     \end{subequations}
     We term $\KK(\cdot)$ the \textit{KR operator} and $\T(s)$ the \textit{accompanying matrix} of $\Lset(s)$. From \eqref{KRdef1},  $\mathcal T (s)$ provides an equivalent vector of currents 
\begin{equation}
\label{eq:I_tilde}
\tilde I_{b}^{dq}(s)= \T(s)I_{\ell}^{dq}(s)
\end{equation}
to be injected in the boundary nodes. 

The matrix $\Lset_{red}(s)=\KK\left(\Lset(s) \right)$ is still the Laplacian of a directed graph $\GG_{red}=\left( \VV_b, \EE_{red}, W_{red}(s)\right) $ \cite{dhople2014synchronization} that is uniquely defined (up to the orientation of edges, which can be arbitrarily chosen) and that has weights $W_{red,ij}(s)=-\left( \Lset_{red}(s)\right)_{ij}$ for $(i,j)\in\EE_{red}$. $\GG_{red}$ is called the \textit{Kron reduced} graph. 

Note that, for given loads $I_{\ell}^{dq}(t)$ and voltages $V^{dq}_b(t)$, $t\geq0$, if initial states of DGUs and lines are zero, then currents $I_b^{dq}(t)$ computed through \eqref{I=LV} and \eqref{KRdef1} are identical at all times. For this reason, \eqref{KRdef1}, along with the graph $\GG_{red}$, can be termed \emph{instantaneous} KR \cite{caliskan2012kron,dhople2014synchronization}.
     
\begin{rmk}
A key issue is to understand when weights $W_{red,ij}(s)$ can be written as in \eqref{ConstitutiveRel} replacing $R_{ij}$ and $L_{ij}$ with suitable parameters $\tilde R_{ij}$ and $\tilde L_{ij}$. It has been shown that this is guaranteed only under special assumptions, for instance if original lines are homogeneous, i.e. $\frac{R_{e_1}}{L_{e_1}}=\frac{R_{e_2}}{L_{e_2}}$, $\forall e_1, e_2 \in \EE$ \cite{caliskan2012kron}. In this case, one also has $\tilde R_{ij}>0$ and $\tilde L_{ij}>0$.
\hspace{68mm}$\blacksquare$\end{rmk}

Next, we introduce AC-KR. Since the three-phase RL line network is balanced and no zero-sequence signals are present, it can be split into three independent and identical single-phase circuits. Each one is associated with the ``single-phase'' directed graph $\GG^{sp}=(\VV, \EE, W^{sp}(s))$ where transfer functions $W_{ij}^{sp}(s)=1/(R_{ij}+sL_{ij})$ represent, independently of the phase $\star=\{a,b,c\}$, the relation between $\mathscr L[V_i^{\star}(t)-V_j^{\star}(t)]$ and $\mathscr  L[I_{ij}^{\star}(t)]$. When the network is in PSSS with frequency $\omega_0$, then $V_i^{\star}(t)=A_i^{\star}\sin(\omega_0t+\phi^{\star}_i)$ and $I_i^{\star}(t)=B_i^{\star}\sin(\omega_0t+\gamma^{\star}_i)$, $\forall i \in \VV$. Moreover, we can associate each sinusoid $V_i^{\star}(t)$ to the corresponding phasor $\vec{V}_i^{\star}=A_i^{\star}\mathrm \exp{( \mathrm i \phi_i^{\star})}$. Current phasors $\vec{I}_i^{\star}$ are defined analogously. Let us define vectors $\vec{V}^{\star}=[\vec{V}_1^{\star},\vec{V}^{\star}_2, ..., \vec{V}^{\star}_n ] ^{T}$ and $\vec{I}^{\star}=[\vec{I}_1^{\star},\vec{I}^{\star}_2, ..., \vec{I}^{\star}_n ]^{T}$, $\star \in \lbrace a,b,c \rbrace$. In PSSS with frequency $\omega_0$, the relation between nodal currents and nodal voltages, for each phase, is given by: $\vec{I}^{\star}=\Lset^{AC}\cdot \vec{V}^{\star}$, $\star \in \lbrace a,b,c \rbrace$, where $\Lset^{AC}_{ij}=-1/Z_{ij}$, if $(i,j)\in\EE$, and $\Lset^{AC}_{ii}=\sum_{j\in\NN_i}1/Z_{ij}$. In particular, $\Lset^{AC}=\Lset(0)$ by construction. 

\begin{defn}
\label{def:AC-KR}
Let  $\vec{V}^{\star}$, $\vec{I}^{\star}$ and $\Lset^{AC}$ be partitioned into boundary and internal components as in \eqref{I=LV} and assume $\Lset_{\ell\ell}^{AC}$ is invertible\footnote{Conditions for the invertibility of $\Lset_{\ell\ell}^{AC}$ have been studied in \cite{luo2014spatiotemporal}.}. AC-KR is given by the graph $\GG_{red}^{AC}=\left( \VV_b, \EE_{red}, W_{red}^{AC} \right)$ associated to $\KK(\Lset^{AC})$ (up to the orientation of edges, which can be arbitrarily chosen) and
\begin{equation}
\label{eq:Ib_fas}
\vec{I}^{\star}_b = \KK(\Lset^{AC}) \vec{V}^{\star}_b-\T^{AC}\vec{I}^{\star}_{\ell} 
\end{equation}
\begin{equation}
\label{eq:tau_AC}
\T^{AC} = -\Lset_{b\ell}^{AC}(\Lset^{AC}_{\ell\ell})^{-1}.
\end{equation}
\end{defn}
We highlight that \eqref{eq:tau_AC} corresponds to \eqref{KRdef3} when $\Lset(s)$ is replaced by $\Lset^{AC}$. Furthermore, \eqref{eq:Ib_fas} provides a relation analogous to \eqref{KRdef1}.

      \subsection{Hybrid KR}
\label{subsec:HybridKR}
     We extend the application of AC-KR to three-phase electrical variables not necessarily in PSSS.
     \begin{defn}
     	The approximate Kron reduced graph $\GG_{red}^{\AAA}$ of the original network $\GG$ is obtained by:
     	\begin{enumerate}[1)]
     		\item computing $\Lset_{red}^{AC}=\KK\left( \Lset^{AC}\right)$ and the associated directed graph $\GG_{red}^{AC}=\left( \VV_b, \EE_{red}, W_{red}^{AC} \right)$;
     		\item setting $\GG_{red}^{\AAA}=\left( \VV_b, \EE_{red}, W_{red}^{\AAA}(s) \right)$ where, for $(i,j)\in\EE_{red}$
     		\begin{equation}
\label{eq:Wred_hKR}
     		W_{red,ij}^{\AAA}(s) =\frac{1}{\tilde Z_ {ij}+\tilde L_{ij}s}, 
     		\end{equation}
     	\begin{equation}
\label{eq:star}
     		\tilde Z_{ij}=-\left(\Lset_{red}^{AC}\right)_{ij}^{-1}, \tilde L_{ij}=\frac{1}{\omega_0}\text{Im}\left(\tilde Z_{ij} \right)
     		\end{equation}
     	\end{enumerate}
     \end{defn}
     In other words, in $\GG_{red}^{\AAA}$, line impedances have still the dynamics \eqref{ConstitutiveRel} but resistances $\tilde R_{ij}=\tilde Z_{ij}-\mathrm i\omega_0 \tilde L_{ij}$ and inductances $\tilde L_{ij}$ are those predicted by AC-KR. Note also that, by construction, graphs associated to $\Lset_{red}^{AC}$ and $\Lset_{red}(s)$ have the same set of undirected edges. Therefore, choosing the same orientation, the set of directed edges of $\GG_{red}$ and $\GG_{red}^{AC}$ can be made identical (this is why they have been both denoted with $\EE_{red}$).
     \begin{rmk}
     	AC-KR does not guarantee positivity of the reduced resistances and inductances 	\cite{dhople2014synchronization}. Based on the results of \cite{caliskan2012kron}, one expects that negative values can occur if the time constants of the original lines are spread in a wide range. In microgrids, however, electrical lines are usually similar and so are their time constants. 
\hspace{98mm}$\blacksquare$
     \end{rmk} 
     In order to complete the KR procedure and specify a formula analogous to \eqref{KRdef1}, it remains to be seen how to compute equivalent boundary currents due to the effects of internal currents that have been eliminated.
 \begin{defn}
\label{def:hKR}
hKR is given by the network $\GG_{red}^{\AAA}$ and the relation
	\begin{equation}
     		\label{hKRa}
     		I_b^{dq}(s)=\Lset_{red}^{\AAA}(s)V_b^{dq}(s)-\tilde I_b^{dq}(s)
              \end{equation} 
     		where $\Lset_{red}^{\AAA}(s)$ is the Laplacian of $\GG_{red}^{\AAA}$ and $\tilde I_b^{dq}(s)$ is defined in \eqref{eq:I_tilde}.
		\end{defn}
In other words, in hKR line dynamics is approximated as in \eqref{eq:Wred_hKR} but internal currents are mapped into boundary currents with no approximation.

	The next result characterizes asymptotic behaviors preserved by hKR.
	\begin{prop}
          \label{prop:as_eq}
Consider the network $\GG$ and assume parameters $\tilde R_{ij}$ and $\tilde L_{ij}$ obtained in hKR are strictly positive. If $I_{\ell}^{abc}$ and $V_b^{abc}$ converge to PSSS with angular frequency $\omega_0$, then the \textit{asymptotic} behavior of $I_b^{dq}$ computed from \eqref{I=LV} is the same as when $I_b^{dq}$ is computed from \eqref{hKRa}.
	\end{prop}
	\begin{proof}
Due to the equivalence of the electrical quantities of the network in \emph{abc} and \emph{dq} reference frames, we use \emph{abc} coordinates. Since the three-phase RL line network is balanced and no zero-sequence signals are present, we make reference to the single-phase directed graph $\GG^{sp}=(\VV,\EE, W^{sp}(s))$, introduced for illustrating AC-KR. The instantaneous relation between $I^{\star}(s)$ and $V^{\star}(s)$, $\star \in \lbrace a,b,c \rbrace$, is given by: $I^{\star}(s)=\Lset^{sp}(s)V^{\star}(s)$, where $\Lset^{sp}(s)$ is the Laplacian matrix of $\GG^{sp}$. After applying hKR, the Laplace transforms of phase $\star \in \lbrace a,b,c \rbrace$ of signals $I_{b}^{abc}(t)$, $V_{b}^{abc}(t)$ and $I_{\ell}^{abc}(t)$ are related by
$$I^{\star}_b(s)=\Lset_{red}^{sp^{\AAA}}(s)V_b^{\star}(s)-\T^{sp}(s)I_{\ell}^{\star}(s)$$
where $\T^{sp}(s)$ is given by \eqref{KRdef3}, computed with respect to $\Lset^{sp}(s)$ and $\Lset_{red}^{sp^{\AAA}}$ is the Laplacian of the single-phase approximate Kron reduced graph $\GG^{sp^{\AAA}}_{red}=(\VV_b,\EE_{red},W^{sp^{\AAA}}_{red}(s))$, with weights $W_{red,ij}^{sp^{\AAA}}(s)=1/(s\tilde L_{ij}+\tilde R_{ij})$. We note that, by construction, the poles of all entries of $\Lset_{red}^{sp^{\AAA}}$ have strictly negative real parts. 
Moreover, also the poles of all entries of $\T^{sp}(s)$ have strictly negative real parts. This can be shown as follows. We recall that the RL network associated to $\GG^{sp}$ is asymptotically stable because it is strictly passive, and, by construction, \eqref{KRdef1} preserves the BIBO stability property between inputs $V_b^{\star}$ and $I^{\star}_{\ell}$ and output $I_b^{\star}$. Hence, in experiments where $V_b^{\star}=0$ and only one load current $I_{\ell_i}^{\star}(t)$ at time is excited, if $I_{\ell_i}^{\star}(t)$ is bounded, also all elements of $I_b^{\star}(t)$ are bounded. By the superposition principle, this shows that all elements of $\T^{sp}(s)$ are BIBO stable (and hence asymptotically stable) SISO transfer functions.

Let us consider the case in which all elements of inputs $V_b^{\star}$ and $I^{\star}_{\ell}$ are sinusoids with angular frequency $\omega_0$. By the frequency response theorem \cite{desoer1984basic}, each element of $I_{b}^{\star}(t)$ tends to a sinusoid. Moreover, the relation between the phasors $\vec{V}_b^{\star}$, $\vec{I}_b^{\star}$ and $\vec{I}_{\ell}^{\star}$ is given by: $\vec{I}_b^{\star}=\Lset_{red}^{sp^{\AAA}}(\mathrm i \omega_0)\vec{V}_b^{\star}+\T^{sp}(\mathrm i \omega_0)\vec{I}_{\ell}^{\star}$, that is \eqref{eq:Ib_fas}.

Since the asymptotic quantities in \emph{abc} coordinates of the original and hybrid Kron reduced network are identical, then they coincide also in \emph{dq} coordinates. Similar arguments can be also applied to the case when $V_b^{\star}$ and $I_{\ell}^{\star}$ are not sinusoids, but asymptotically reach PSSS with frequency $\omega_0$.
\end{proof}

	\subsection{Numerical examples}
\label{subsec:ne}
    Consider the hKR of the three-phase network in Figure \ref{Fig:NetworkOriginal}, composed of three ideal voltage sources and balanced RL lines connecting the generators to a purely resistive load. The three-phase voltage generators are balanced and have angular frequency $\omega_0=2\pi50\,\mbox{rad/s}$. In this very simple case, hKR is a Y$-\Delta$ transformation. The reduced network is shown in Figure \ref{Fig:NetworkReduced}. 
    The corresponding voltage generators are identical for the original and reduced networks. Hence, from Proposition \ref{prop:as_eq}, we expect $I_b$ in the original and reduced models to be the same, once PSSS is reached.
 \begin{figure}
    	\centering
    	\begin{subfigure}[htb]{0.5\textwidth}
    		\centering
    		\includegraphics[scale=0.5]{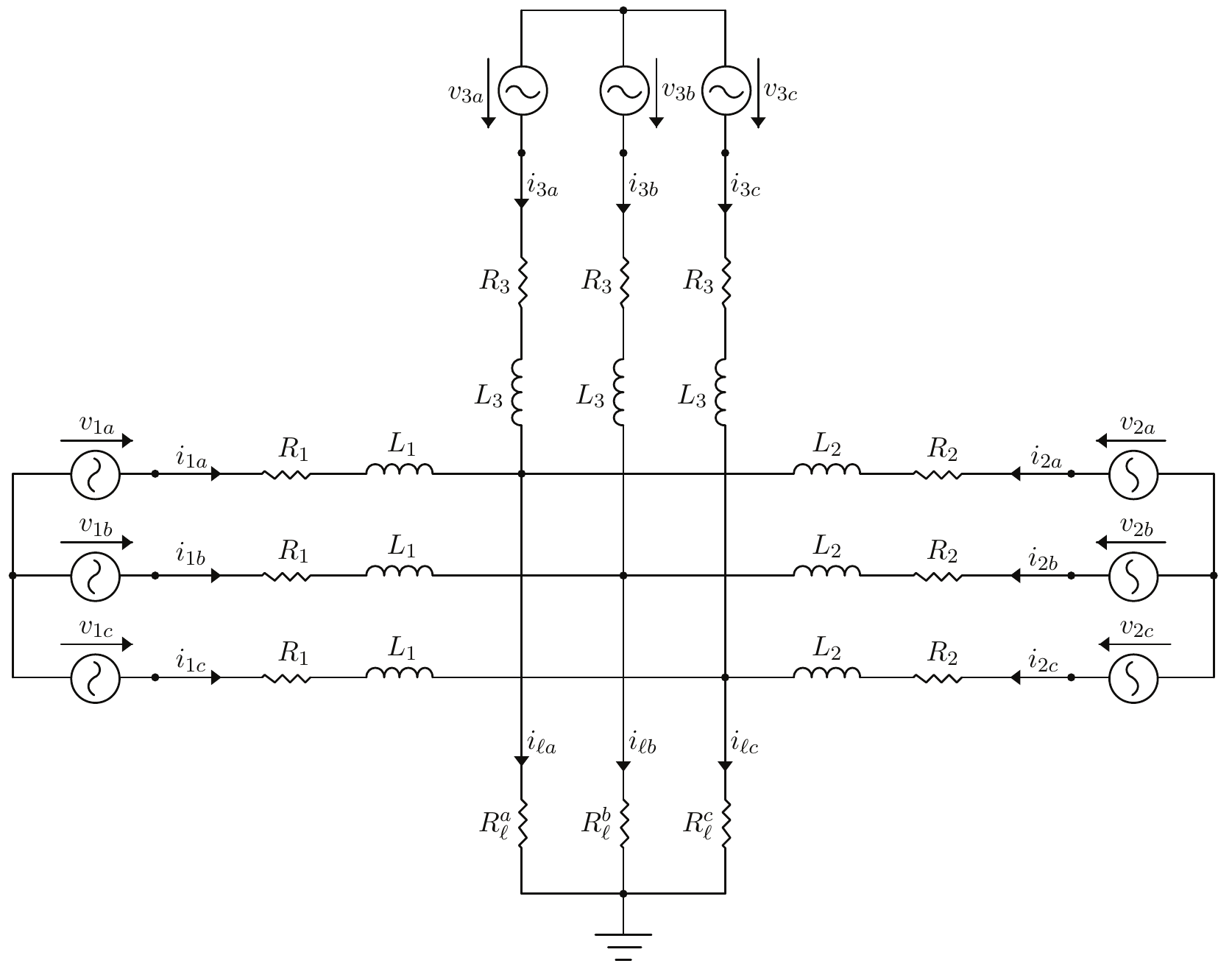}
    		\caption{Original network.}
    		\label{Fig:NetworkOriginal}
    	\end{subfigure}
    	\begin{subfigure}[htb]{0.5\textwidth}
    		\centering
    		\vspace{5mm}
    		\includegraphics[scale=0.5]{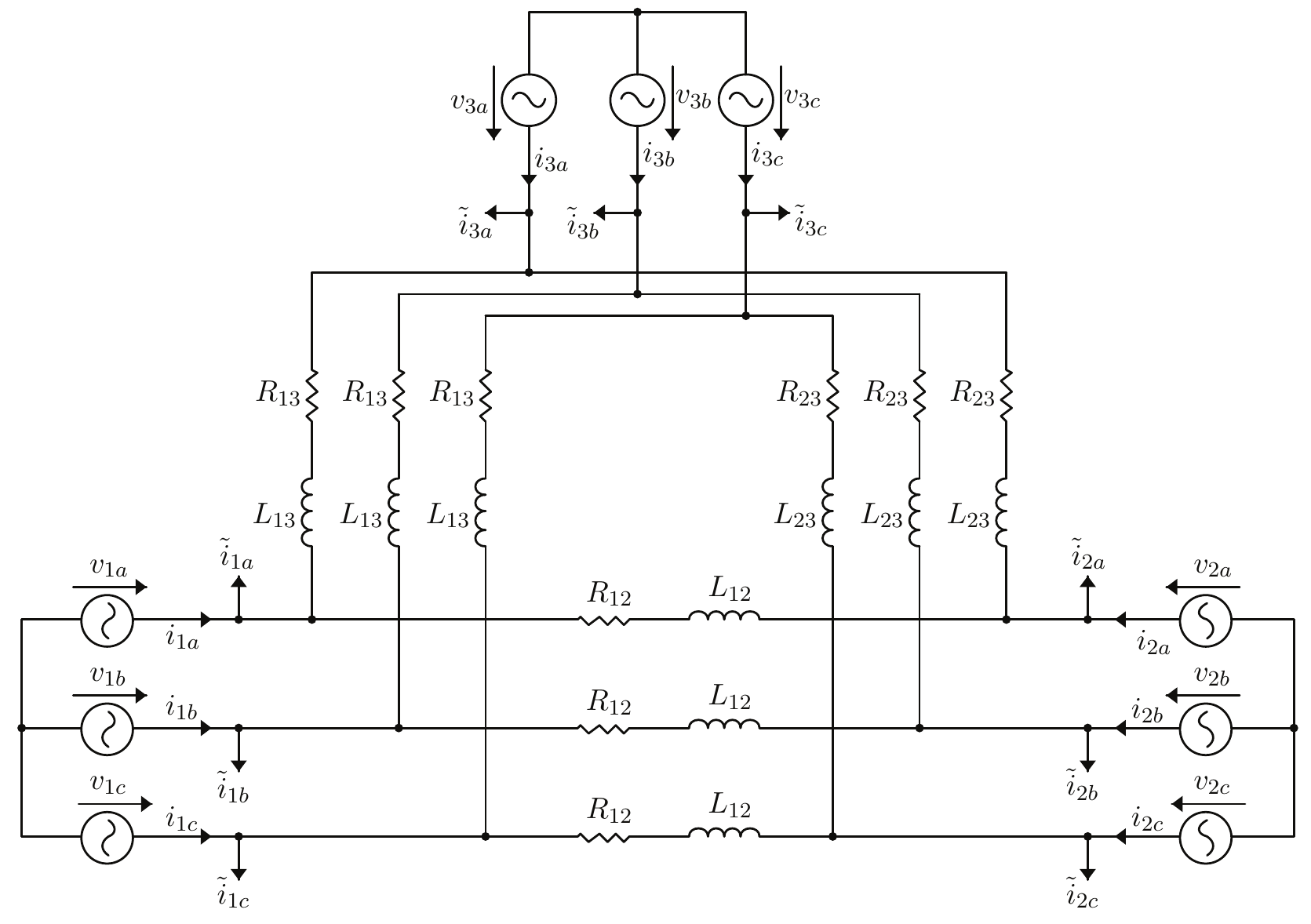}
    		\caption{Reduced model obtained from hKR.}
    		\label{Fig:NetworkReduced}
    	\end{subfigure}
    	\caption{Original and reduced networks.}
    	\label{fig:21nodes_performance}
    \end{figure}
    The lines and loads parameters, as well as the voltages of generators,  are collected in Tables \ref{Appendix:3FTable:OriginalParam}, \ref{Appendix:3FTable:ReducedParam}, \ref{Appendix:3FTable:Loads}, \ref{Appendix:TableNonlinearload} and \ref{Appendix:TableVoltages} in Appendix \ref{app:ne}.
    
    \textit{Example 1 - Linear unbalanced load.$\;$}
    In this example, we suppose that the resistive load is unbalanced and compare the original circuit with the hybrid Kron reduced model. 
	Figure \ref{Fig:Ex1} shows the output current of phase $a$ of generator 1. 
	This example highlights that hKR ensures asymptotic equivalence, even if the load is unbalanced. In the presented case, the lines are mainly resistive, so equivalence is achieved almost immediately. 
\begin{figure}
			\centering
\includegraphics[scale=0.2]{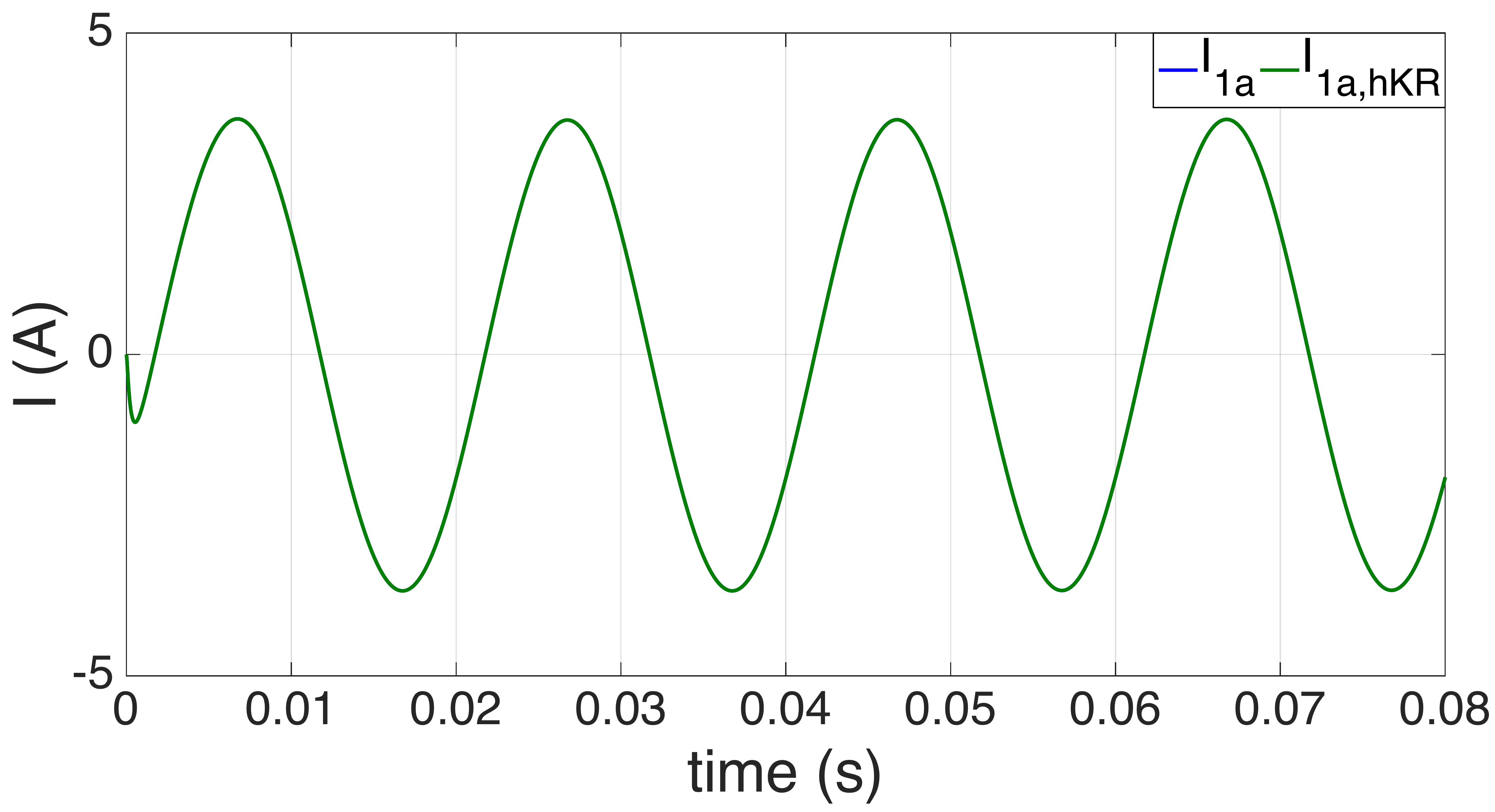}
		\caption{Example 1 - Comparison between original (blue line) and reduced (green line) phase $a$ currents of generator 1 with unbalanced load.}
		\label{Fig:Ex1}
\end{figure}

\textit{Example 2 - Nonlinear load.$\;$}
	In this second example, we consider the same network of Example 1, but with the resistive load replaced by a six-pulse bridge rectifier. We compare the currents at boundary nodes of the original circuit with those of the reduced model. For the sake of simplicity, we concentrate on the output currents of generator 1, phase $a$, shown in Figure \ref{Fig:2}. Even if the load currents are highly distorted, the difference between boundary currents of the original and reduced network goes asymptotically to zero, as shown in Figure \ref{Fig3:HybridNL3}. This is because Proposition 1 provides a sufficient, but not necessary, condition for asymptotic equivalence of the boundary currents. Analogous results are obtained for all the other phases and generators.
\begin{figure}
\centering
\begin{subfigure}[htb]{0.48\textwidth}
			\centering
			\includegraphics[width=1\textwidth, height = 0.2\textheight]{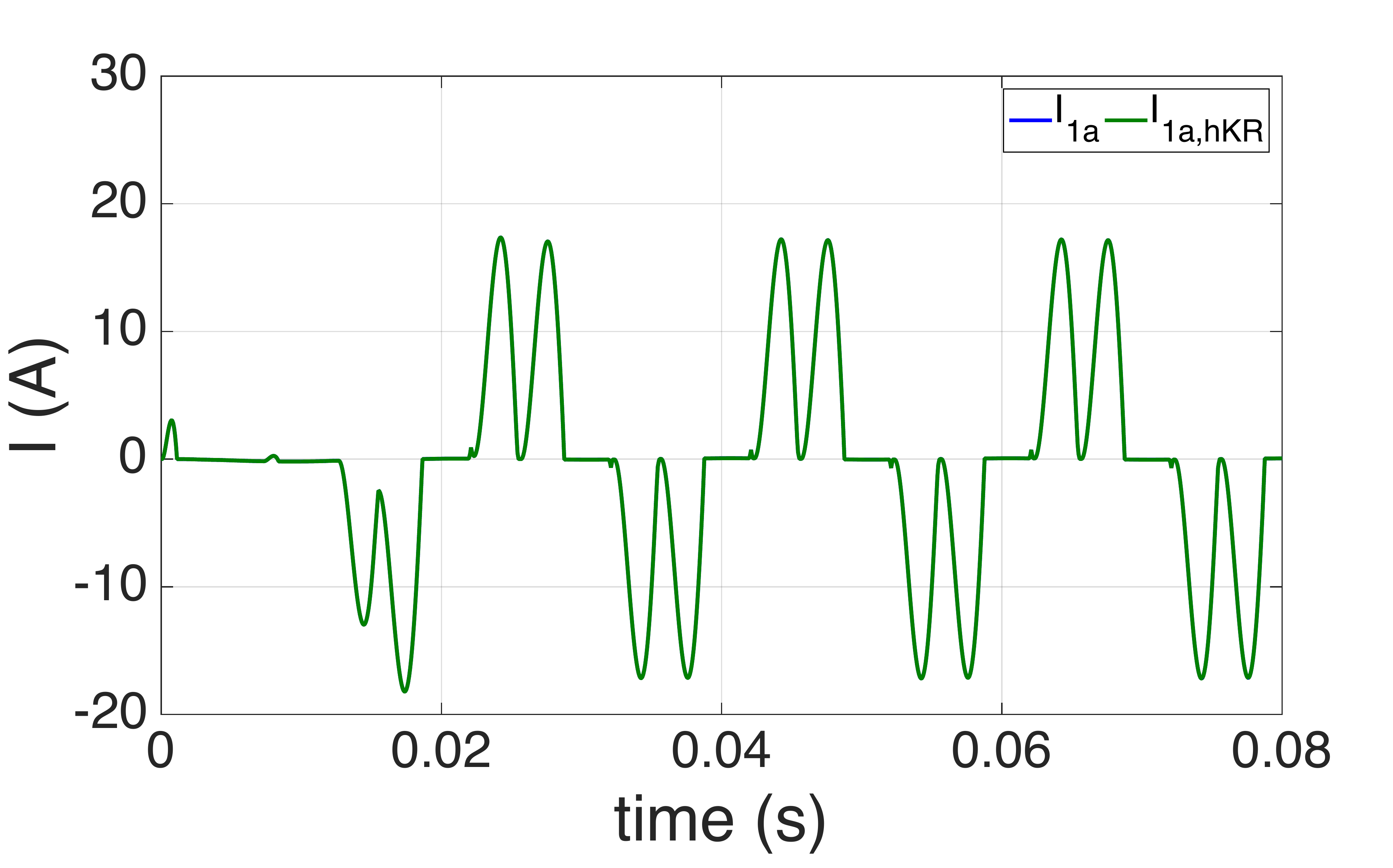}
			\caption{Output currents of phase $a$.}
			\label{Fig:2}
		\end{subfigure}
                \begin{subfigure}[htb]{0.48\textwidth}
			\centering
			\includegraphics[width=1\textwidth, height = 0.2\textheight]{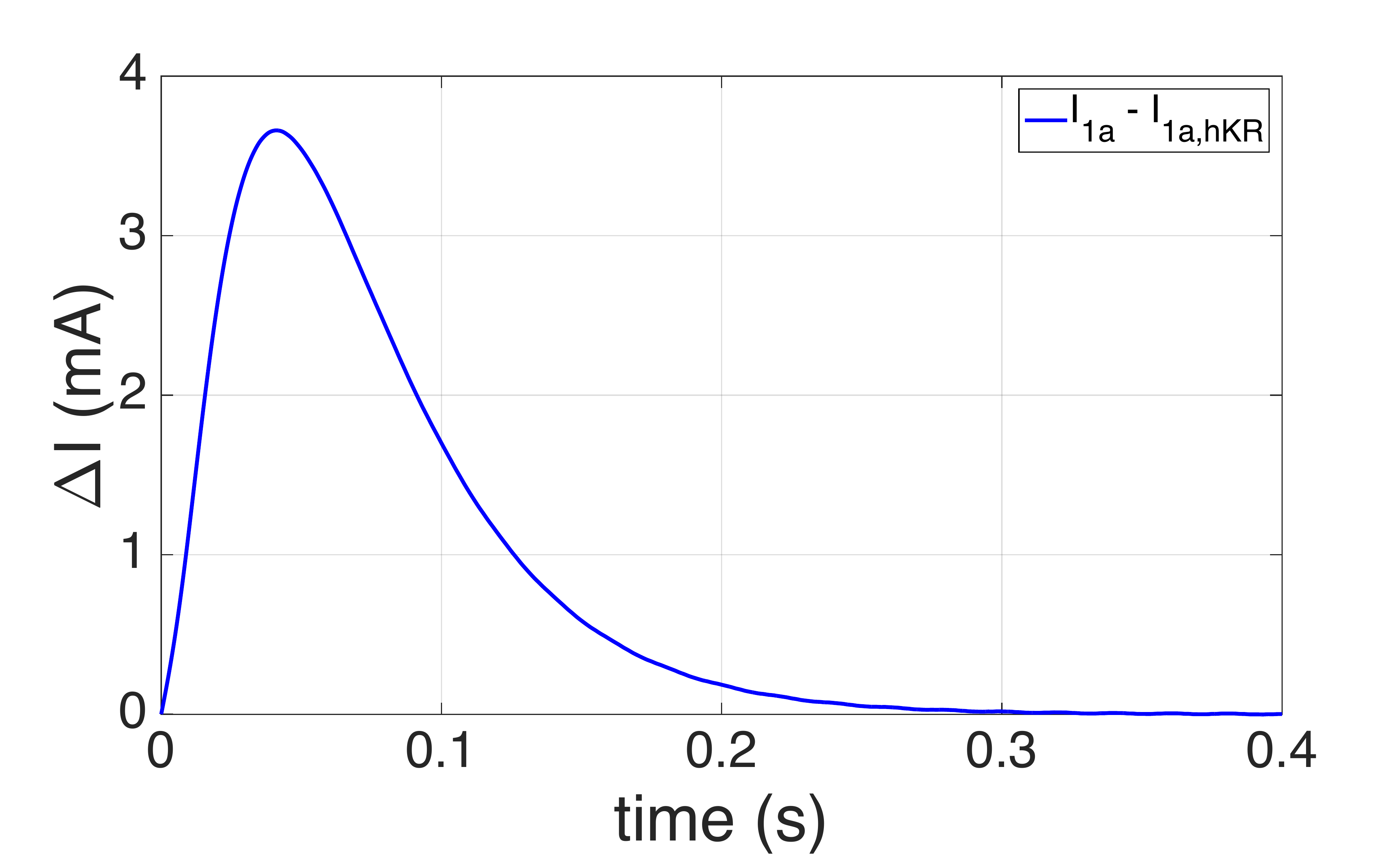}
			\caption{Error between boundary currents.}
			\label{Fig3:HybridNL3}
		\end{subfigure}
		\caption{Example 2 - Comparison between original and reduced currents of phase $a$ of generator 1 with nonlinear load.}
		\label{Fig3:HybridNL}
	\end{figure}

\section{PnP control of ImGs}
\label{sec:PnPctrl}
In this Section we briefly summarize the PnP algorithm in \cite{riverso2015plug} for designing a decentralized control architecture guaranteeing voltage and frequency stability in an ImG. Then, we will show how to generalize it to the case of ImGs with arbitrary topologies. The method in \cite{riverso2015plug} assumes a load-connected topology and RL lines with strictly positive resistances and inductances. Local regulators $\CC_i$ use measurements of the voltage $V^{dq}_i$ at PCC (see Figure \ref{fig:schema1DGU}) and the current $I^{dq}_{ti}$ to control the voltage $V_{ti}^{dq}$ at the VSC $i$ so as to make $V_{i}^{dq}$ track a reference signal. Furthermore, each controller is composed of a matrix gain and an integral action on the $d$ and $q$ components of the tracking error.
When a DGU (say DGU $i$) wants to join the network (e.g. DGU 3 in Figure \ref{fig:Ex_plugin_req}), it issues a plug-in request to its future neighbors, i.e. DGUs $j\in\NN_i$ (see, for example, DGUs 2 and 4 in Figure \ref{fig:Ex_plugin_req}). DGU $i$ then solves the Linear Matrix Inequality (LMI) problem (19) in \cite{riverso2015plug}, which depends only upon the parameters of the lines $ij$. If feasible, the optimization problem produces a controller $\CC_i$, along with a local and structured Lyapunov function that can be used for certifying stability of the whole ImG. Since also DGUs $j\in\NN_i$ will have a new neighbor, they must update their controller $\CC_j$ by tacking into account the parameters of the new line $ji$. This is done by solving an LMI problem analogous to the one solved by DGU $i$. 
If one of the above LMI problems is unfeasible, plug-in of DGU $i$ is denied. Otherwise, 
DGU $i$ can be connected and stability of the whole ImG can be certified using the sum of the computed local Lyapunov functions. 

Unplugging of a DGU (say DGU $m$) follows a similar procedure: as line $mk$, $k\in\NN_m$ will be disconnected from the corresponding DGU $k$, all controllers $\CC_k$ must be successfully redesigned before allowing the disconnection. Following a similar reasoning, also the possibility of changing a parameter of line $ij$ (or add a new line $ij$), must be first tested by successfully designing controllers $\CC_i$ and $\CC_j$ through suitable LMIs. 
While line changes are not common in ImGs, they could happen in the hybrid Kron reduced network because of the addition/removal of a load node in the original ImG. Examples of this phenomenon are provided in the next Section. We also highlight that PnP controllers in \cite{riverso2015plug} can be enhanced with pre-filters of reference signals and compensators of (measured) load currents $I_{Li}$ represented in Figure 2.\footnote{These enhancements are optional and they will be not used in the simulations in Section \ref{sec:Sim21bn}.}

When the design of PnP controllers is conducted with reference to the hybrid Kron reduced model, according to the electrical scheme in Figure \ref{fig:schema1DGU}, current $\tilde I_{b_k}^{dq}$ in \eqref{hKRa} can be lumped into current $I_{Li}^{dq}$, if $b_k = i$. In other words, the effect of load nodes is mapped into additional contributions to DGU loads. This is, however, not critical because DGU loads are treated as disturbances by PnP controllers.

  \begin{figure}
		\centering
		\includegraphics[scale=0.5]{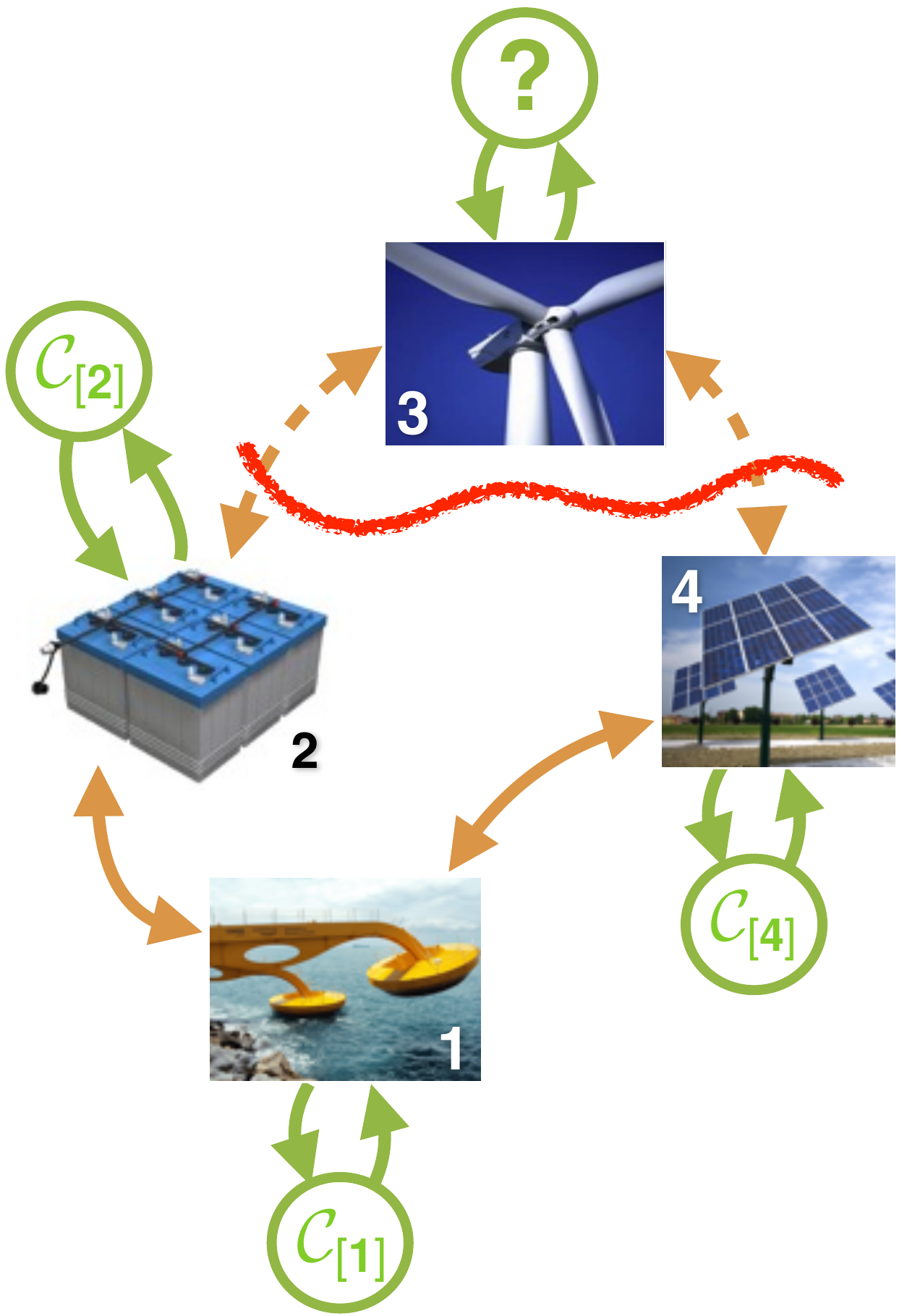}
\caption{Example of plug-in request issued by DGU 3 to its future neighbors (i.e. DGUs 2 and 4).}
		\label{fig:Ex_plugin_req}
	\end{figure}

\subsection{PnP design for general ImGs}
\label{sec:PnPdesfen}
Let us consider an ImG with arbitrary topology represented by the directed connected graph $\GG$ with node set $\VV$. When a DGU or a load node (say node $i$) wants to plug in, the first step consists in updating $\GG$ accordingly, so obtaining the graph $\GG^{new}$. Then, hKR is applied to $\GG^{new}$ for obtaining the reduced ImG $\GG_{red}^{\AAA,new}$ with load-connected topology. If some resistances or inductances of reduced lines are negative, the plugging-in of node $i$ is denied, as one of the assumption of the PnP algorithm in \cite{riverso2015plug} is not fulfilled. Otherwise, one compares the reduced graphs $\GG_{red}^{\AAA}$ (associated to $\GG$) and $\GG_{red}^{\AAA,new}$ for finding the set $\UU\subseteq\VV_b$ of DGUs that have new neighbors or that are connected to lines whose impedance has changed. The LMI problem (19) in \cite{riverso2015plug} is then solved for all DGUs $j\in\UU$ (and also for $j=i$, if node $i$ is a DGU), hence producing new controllers $\CC_j$. If no LMI is infeasible, controllers in the original ImG are updated and connection of node $i$ is allowed.

Unplugging of a node can be performed in a similar way.
\begin{rmk}
\label{rmk:new}
According to the above algorithm, hKR is performed in a centralized fashion every time there is a change in the network topology. This is in contrast with PnP design, whose main feature is to avoid any centralized computations. In the future, we will study how to perform hKR in a distributed fashion. Notably, one can develop this generalization according to the iterative KR procedure \cite{dorfler2013kron} proposed for both AC networks in PSSS and resistive networks. In a similar spirit, we will study how to avoid the centralized computation of the set $\UU$ by exploiting existing distributed algorithms for path-finding over directed graphs. 
\hspace{75mm}$\blacksquare$
\end{rmk}

	\section{Simulation of a 21-bus network}
          \label{sec:Sim21bn}
In this Section, we assess the capability of PnP control and hKR to deal with networks characterized by
complex topologies. In particular, we use a network derived from the
top half of IEEE 37 topology \cite{feeders2011ieee} and identify generation nodes and
loads as in \cite{bolognani2013distributed}.\\
The simulation, whose duration is 13 seconds, has been performed
in PSCAD.

	\subsection{ImG topology}
	The lack in literature of standard test benches for microgrids has led many authors to use IEEE test feeders \cite{feeders2011ieee} as workbenches for microgrids testing. These networks have been originally developed as standards for power flows studies in high-voltage radial grids. When the IEEE test feeders are applied to microgrids, only the network topologies are preserved, while line parameters and loads are changed. 
	\begin{figure}
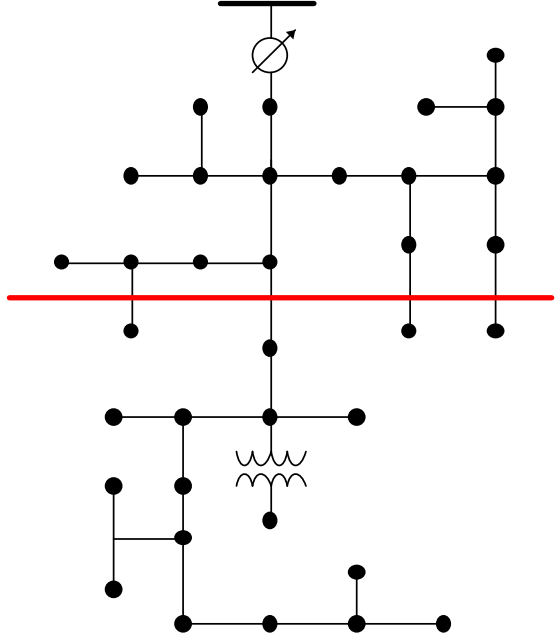

		\centering
		\begin{pgfpicture}{0cm}{0cm}{213pt}{265pt}
\pgfsetxvec{\pgfpoint{1pt}{0pt}}
\pgfsetyvec{\pgfpoint{0pt}{1pt}}
\pgfsetroundjoin 
\pgfsetroundcap
\pgftranslateto{\pgfxy(0,265)}
\begin{pgfmagnify}{1}{-1}
\definecolor{layer0}{rgb}{0.0,0.0,0.0}
\definecolor{layer1}{rgb}{0.0,0.0,0.5}
\definecolor{layer2}{rgb}{1.0,0.0,0.0}
\definecolor{layer3}{rgb}{0.0,0.5,0.5}
\definecolor{layer4}{rgb}{1.0,0.78,0.0}
\definecolor{layer5}{rgb}{0.5,1.0,0.0}
\definecolor{layer6}{rgb}{0.0,1.0,1.0}
\definecolor{layer7}{rgb}{0.0,0.5,0.0}
\definecolor{layer8}{rgb}{0.6,0.8,0.2}
\definecolor{layer9}{rgb}{1.0,0.08,0.58}
\definecolor{layer10}{rgb}{0.71,0.61,0.05}
\definecolor{layer11}{rgb}{0.0,0.5,1.0}
\definecolor{layer12}{rgb}{0.88,0.88,0.88}
\definecolor{layer13}{rgb}{0.64,0.64,0.64}
\definecolor{layer14}{rgb}{0.37,0.37,0.37}
\definecolor{layer15}{rgb}{0.0,0.0,0.0}
\color{layer0}
\pgfsetlinewidth{0.65pt}
\pgfellipse[fillstroke]{\pgfxy(103.5,45.0)}{\pgfxy(2.5,0)}{\pgfxy(0,3.0)}
\pgfellipse[fillstroke]{\pgfxy(77.5,45.0)}{\pgfxy(2.5,0)}{\pgfxy(0,3.0)}
\pgfellipse[fillstroke]{\pgfxy(77.5,71.0)}{\pgfxy(2.5,0)}{\pgfxy(0,3.0)}
\pgfellipse[fillstroke]{\pgfxy(103.5,71.0)}{\pgfxy(2.5,0)}{\pgfxy(0,3.0)}
\pgfellipse[fillstroke]{\pgfxy(51.5,71.0)}{\pgfxy(2.5,0)}{\pgfxy(0,3.0)}
\pgfellipse[fillstroke]{\pgfxy(129.5,71.0)}{\pgfxy(2.5,0)}{\pgfxy(0,3.0)}
\pgfellipse[fillstroke]{\pgfxy(155.5,71.0)}{\pgfxy(2.5,0)}{\pgfxy(0,3.0)}
\pgfsetlinewidth{2.0pt}
\pgfline{\pgfxy(85,6)}{\pgfxy(120,6)}
\pgfsetlinewidth{0.65pt}
\pgfline{\pgfxy(104.0,32.0)}{\pgfxy(104.0,44.0)}
\pgfline{\pgfxy(78.0,71.0)}{\pgfxy(78.0,45.0)}
\pgfline{\pgfxy(104.0,45.0)}{\pgfxy(104.0,71.0)}
\pgfline{\pgfxy(104.0,71.0)}{\pgfxy(78.0,71.0)}
\pgfline{\pgfxy(78.0,71.0)}{\pgfxy(52.0,71.0)}
\pgfline{\pgfxy(104.0,71.0)}{\pgfxy(130.0,71.0)}
\pgfline{\pgfxy(130.0,71.0)}{\pgfxy(156.0,71.0)}
\pgfellipse[fillstroke]{\pgfxy(103.5,103.5)}{\pgfxy(2.5,0)}{\pgfxy(0,2.5)}
\pgfellipse[fillstroke]{\pgfxy(77.5,103.5)}{\pgfxy(2.5,0)}{\pgfxy(0,2.5)}
\pgfellipse[fillstroke]{\pgfxy(51.5,103.5)}{\pgfxy(2.5,0)}{\pgfxy(0,2.5)}
\pgfellipse[fillstroke]{\pgfxy(25.5,103.5)}{\pgfxy(2.5,0)}{\pgfxy(0,2.5)}
\pgfellipse[fillstroke]{\pgfxy(51.5,129.5)}{\pgfxy(2.5,0)}{\pgfxy(0,2.5)}
\pgfellipse[fillstroke]{\pgfxy(103.5,136.0)}{\pgfxy(2.5,0)}{\pgfxy(0,3.0)}
\pgfellipse[fillstroke]{\pgfxy(103.5,162.0)}{\pgfxy(2.5,0)}{\pgfxy(0,3.0)}
\pgfellipse[fillstroke]{\pgfxy(136.0,162.0)}{\pgfxy(3.0,0)}{\pgfxy(0,3.0)}
\pgfellipse[fillstroke]{\pgfxy(103.5,201.0)}{\pgfxy(2.5,0)}{\pgfxy(0,3.0)}
\pgfellipse[fillstroke]{\pgfxy(71.0,162.0)}{\pgfxy(3.0,0)}{\pgfxy(0,3.0)}
\pgfellipse[fillstroke]{\pgfxy(45.0,162.0)}{\pgfxy(3.0,0)}{\pgfxy(0,3.0)}
\pgfellipse[fillstroke]{\pgfxy(71.0,188.0)}{\pgfxy(3.0,0)}{\pgfxy(0,3.0)}
\pgfellipse[fillstroke]{\pgfxy(71.0,207.5)}{\pgfxy(3.0,0)}{\pgfxy(0,2.5)}
\pgfellipse[fillstroke]{\pgfxy(45.0,188.0)}{\pgfxy(3.0,0)}{\pgfxy(0,3.0)}
\pgfellipse[fillstroke]{\pgfxy(45.0,227.0)}{\pgfxy(3.0,0)}{\pgfxy(0,3.0)}
\pgfline{\pgfxy(45.0,227.0)}{\pgfxy(45.0,188.0)}
\pgfline{\pgfxy(45.0,208.0)}{\pgfxy(71.0,208.0)}
\pgfline{\pgfxy(71.0,208.0)}{\pgfxy(71.0,188.0)}
\pgfline{\pgfxy(71.0,188.0)}{\pgfxy(71.0,162.0)}
\pgfline{\pgfxy(71.0,162.0)}{\pgfxy(45.0,162.0)}
\pgfline{\pgfxy(71.0,162.0)}{\pgfxy(104.0,162.0)}
\pgfline{\pgfxy(104.0,162.0)}{\pgfxy(136.0,162.0)}
\pgfline{\pgfxy(104.0,162.0)}{\pgfxy(104.0,65.0)}
\pgfline{\pgfxy(104.0,104.0)}{\pgfxy(26.0,104.0)}
\pgfline{\pgfxy(52.0,104.0)}{\pgfxy(52.0,130.0)}
\pgfellipse[fillstroke]{\pgfxy(71.0,240.0)}{\pgfxy(3.0,0)}{\pgfxy(0,3.0)}
\pgfellipse[fillstroke]{\pgfxy(103.5,240.0)}{\pgfxy(2.5,0)}{\pgfxy(0,3.0)}
\pgfellipse[fillstroke]{\pgfxy(136.0,240.0)}{\pgfxy(3.0,0)}{\pgfxy(0,3.0)}
\pgfellipse[fillstroke]{\pgfxy(136.0,220.5)}{\pgfxy(3.0,0)}{\pgfxy(0,2.5)}
\pgfellipse[fillstroke]{\pgfxy(168.5,240.0)}{\pgfxy(2.5,0)}{\pgfxy(0,3.0)}
\pgfline{\pgfxy(71.0,208.0)}{\pgfxy(71.0,240.0)}
\pgfline{\pgfxy(71.0,240.0)}{\pgfxy(169.0,240.0)}
\pgfline{\pgfxy(136.0,240.0)}{\pgfxy(136.0,221.0)}
\pgfellipse[fillstroke]{\pgfxy(155.5,97.0)}{\pgfxy(2.5,0)}{\pgfxy(0,3.0)}
\pgfellipse[fillstroke]{\pgfxy(155.5,129.5)}{\pgfxy(2.5,0)}{\pgfxy(0,2.5)}
\pgfellipse[fillstroke]{\pgfxy(188.0,71.0)}{\pgfxy(3.0,0)}{\pgfxy(0,3.0)}
\pgfellipse[fillstroke]{\pgfxy(188.0,97.0)}{\pgfxy(3.0,0)}{\pgfxy(0,3.0)}
\pgfellipse[fillstroke]{\pgfxy(188.0,129.5)}{\pgfxy(3.0,0)}{\pgfxy(0,2.5)}
\pgfellipse[fillstroke]{\pgfxy(188.0,45.0)}{\pgfxy(3.0,0)}{\pgfxy(0,3.0)}
\pgfellipse[fillstroke]{\pgfxy(188.0,25.5)}{\pgfxy(3.0,0)}{\pgfxy(0,2.5)}
\pgfellipse[fillstroke]{\pgfxy(162.0,45.0)}{\pgfxy(3.0,0)}{\pgfxy(0,3.0)}
\pgfline{\pgfxy(156.0,71.0)}{\pgfxy(156.0,130.0)}
\pgfline{\pgfxy(188.0,130.0)}{\pgfxy(188.0,26.0)}
\pgfline{\pgfxy(188.0,45.0)}{\pgfxy(162.0,45.0)}
\pgfline{\pgfxy(188.0,71.0)}{\pgfxy(156.0,71.0)}
\pgfmoveto{\pgfxy(91,175)} 
\pgfcurveto{\pgfxy(92,182)}{\pgfxy(96,182)}{\pgfxy(97,175)}
\pgfstroke
\pgfmoveto{\pgfxy(97,175)} 
\pgfcurveto{\pgfxy(98,182)}{\pgfxy(102,182)}{\pgfxy(104,175)}
\pgfstroke
\pgfmoveto{\pgfxy(104,175)} 
\pgfcurveto{\pgfxy(105,182)}{\pgfxy(109,182)}{\pgfxy(110,175)}
\pgfstroke
\pgfmoveto{\pgfxy(110,175)} 
\pgfcurveto{\pgfxy(111,182)}{\pgfxy(115,182)}{\pgfxy(117,175)}
\pgfstroke
\pgfmoveto{\pgfxy(91,188)} 
\pgfcurveto{\pgfxy(92,182)}{\pgfxy(96,182)}{\pgfxy(97,188)}
\pgfstroke
\pgfmoveto{\pgfxy(97,188)} 
\pgfcurveto{\pgfxy(98,182)}{\pgfxy(102,182)}{\pgfxy(104,188)}
\pgfstroke
\pgfmoveto{\pgfxy(104,188)} 
\pgfcurveto{\pgfxy(105,182)}{\pgfxy(109,182)}{\pgfxy(110,188)}
\pgfstroke
\pgfmoveto{\pgfxy(110,188)} 
\pgfcurveto{\pgfxy(111,182)}{\pgfxy(115,182)}{\pgfxy(117,188)}
\pgfstroke
\pgfline{\pgfxy(104.0,175.0)}{\pgfxy(104.0,162.0)}
\pgfline{\pgfxy(104.0,188.0)}{\pgfxy(104.0,201.0)}
\pgfellipse[stroke]{\pgfxy(103.5,25.5)}{\pgfxy(6.5,0)}{\pgfxy(0,6.5)}
\pgfline{\pgfxy(104.0,19.0)}{\pgfxy(104.0,6.0)}
\pgfline{\pgfxy(97.0,32.0)}{\pgfxy(113.0,16.0)}
\pgfmoveto{\pgfxy(113.0,16.0)}
\pgflineto{\pgfxy(112.29289321881345,19.535533905932738)}
\pgflineto{\pgfxy(109.46446609406726,16.707106781186546)}
\pgfclosepath 
\pgffill 
\color{layer2}
\pgfsetlinewidth{2.0pt}
\pgfline{\pgfxy(6,117)}{\pgfxy(209,117)}
\end{pgfmagnify}
\end{pgfpicture}
		\caption{IEEE 37 Node Test Feeder and the top half network considered in our experiments.}
		\label{Fig6:IEEE37} 
	\end{figure}
	In literature, IEEE 37 network, with 37 nodes, appears to be a
        widespread topology \cite{dorfler2014breaking},
        \cite{bolognani2013distributed}. The IEEE 37 network, shown in
        Figure \ref{Fig6:IEEE37}, has a radial topology, which is quite common in its original context of high/medium-voltage networks. However, microgrids could present more complex topologies (for example, they may have a meshed structure). As a consequence, IEEE 37 appears very simple for PnP control purposes, because it does not contain any loop.
	In order to show the versatility of the approach based on PnP
        control and hKR, we have proposed a new network,
        derived from the top half of IEEE 37 topology in
        Figure \ref{Fig6:IEEE37}. 
	
	The proposed topology is shown in Figure \ref{Fig6:Retemia}. It has 21 nodes, with six DGUs, electrical RL lines having time constants
spread in a wide range, linear R and RL loads, as well as highly
nonlinear and highly inductive loads. Notice that all these loads appear as internal nodes and, without loss of generality, no load $I_{Li}$ directly connected to each PCC is present. Compared to the IEEE 37 network, a switch $SW_1$ has been introduced, allowing the
plugging-in/unplugging of loads at nodes 16, 17 and 18. Moreover, two branches
($e_{18}$ and $e_{19}$) were added and connected to the
microgrid through switches $SW_2$ e $SW_3$ respectively. The edge
$e_{18}$ creates a mesh between DGUs 1, 3 and 4; this enables us to
show that PnP controllers can stabilize also meshed
networks. The edge $e_{19}$ simply changes the impedance between DGUs 3
and 5 (as long as $SW_1$ is closed). Finally, one generation node (21)
and two loads (at nodes 19 and 20) have been introduced, so as to
simulate the plugging-in of a new DGU. The new generation unit is
connected to the microgrid via switch $SW_4$. 
	\begin{figure}
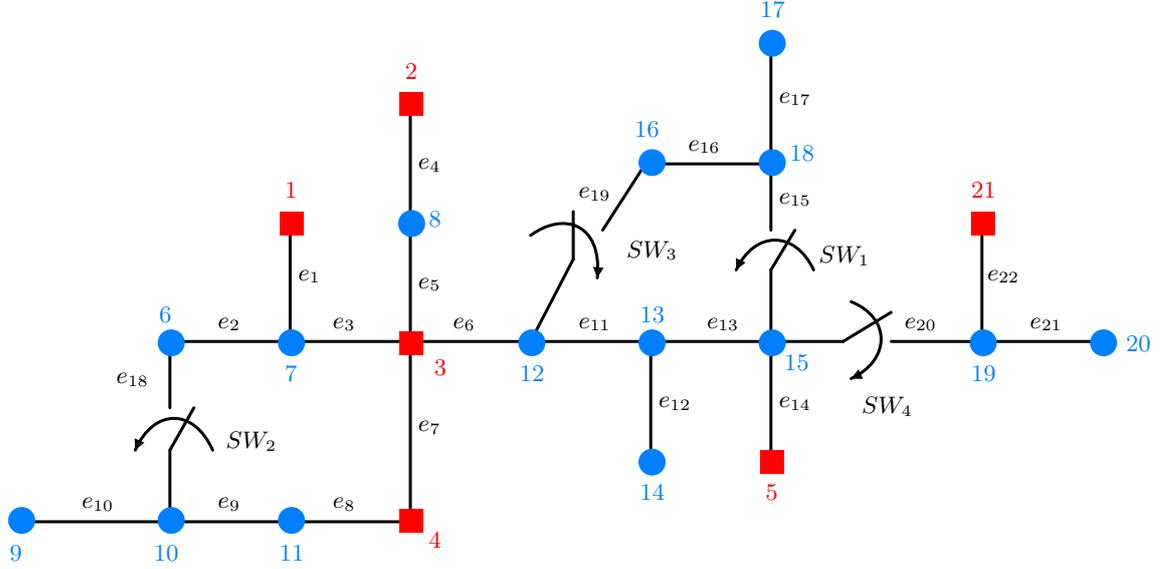

		\centering
		\begin{small}
			\begin{pgfpicture}{0cm}{0cm}{459pt}{252pt}
\pgfsetxvec{\pgfpoint{1pt}{0pt}}
\pgfsetyvec{\pgfpoint{0pt}{1pt}}
\pgfsetroundjoin 
\pgfsetroundcap
\pgftranslateto{\pgfxy(0,252)}
\begin{pgfmagnify}{1}{-1}
\definecolor{layer0}{rgb}{0.0,0.0,0.0}
\definecolor{layer1}{rgb}{0.0,0.0,0.5}
\definecolor{layer2}{rgb}{1.0,0.0,0.0}
\definecolor{layer3}{rgb}{0.0,0.5,0.5}
\definecolor{layer4}{rgb}{1.0,0.78,0.0}
\definecolor{layer5}{rgb}{0.5,1.0,0.0}
\definecolor{layer6}{rgb}{0.0,1.0,1.0}
\definecolor{layer7}{rgb}{0.0,0.5,0.0}
\definecolor{layer8}{rgb}{0.6,0.8,0.2}
\definecolor{layer9}{rgb}{1.0,0.08,0.58}
\definecolor{layer10}{rgb}{0.71,0.61,0.05}
\definecolor{layer11}{rgb}{0.0,0.5,1.0}
\definecolor{layer12}{rgb}{0.88,0.88,0.88}
\definecolor{layer13}{rgb}{0.64,0.64,0.64}
\definecolor{layer14}{rgb}{0.37,0.37,0.37}
\definecolor{layer15}{rgb}{0.0,0.0,0.0}
\color{layer0}
\pgfsetlinewidth{1.125pt}
\pgfline{\pgfxy(10.0,203.0)}{\pgfxy(158.0,203.0)}
\pgfline{\pgfxy(158.0,203.0)}{\pgfxy(158.0,45.0)}
\pgfline{\pgfxy(68.0,135.0)}{\pgfxy(293.0,135.0)}
\pgfline{\pgfxy(293.0,93.0)}{\pgfxy(293.0,23.0)}
\pgfline{\pgfxy(248.0,68.0)}{\pgfxy(293.0,68.0)}
\pgfline{\pgfxy(248.0,180.0)}{\pgfxy(248.0,135.0)}
\pgfline{\pgfxy(113.0,90.0)}{\pgfxy(113.0,135.0)}
\pgfline{\pgfxy(293.0,180.0)}{\pgfxy(293.0,135.0)}
\pgfline{\pgfxy(372.0,90.0)}{\pgfxy(372.0,135.0)}
\pgfline{\pgfxy(372.0,135.0)}{\pgfxy(417.0,135.0)}
\pgfline{\pgfxy(203.0,135.0)}{\pgfxy(219.0,104.0)}
\pgfline{\pgfxy(219.0,104.0)}{\pgfxy(219.0,86.0)}
\pgfline{\pgfxy(230.0,93.0)}{\pgfxy(246.0,68.0)}
\pgfmoveto{\pgfxy(203,95)} 
\pgfcurveto{\pgfxy(215,86)}{\pgfxy(230,90)}{\pgfxy(228,111)}
\pgfstroke
\pgfmoveto{\pgfxy(228.0,111.0)}
\pgflineto{\pgfxy(226.38824542532407,106.8283999243682)}
\pgflineto{\pgfxy(230.3702273156999,107.20763629488019)}
\pgfclosepath 
\pgffill 
\pgfline{\pgfxy(293.0,135.0)}{\pgfxy(320.0,135.0)}
\pgfline{\pgfxy(320.0,135.0)}{\pgfxy(338.0,124.0)}
\pgfline{\pgfxy(338.0,135.0)}{\pgfxy(372.0,135.0)}
\pgfline{\pgfxy(293.0,135.0)}{\pgfxy(293.0,108.0)}
\pgfline{\pgfxy(293.0,108.0)}{\pgfxy(302.0,93.0)}
\pgfmoveto{\pgfxy(280,108)} 
\pgfcurveto{\pgfxy(287,93)}{\pgfxy(302,93)}{\pgfxy(309,108)}
\pgfstroke
\pgfmoveto{\pgfxy(280.0,108.0)}
\pgflineto{\pgfxy(279.8791755813339,103.5294965093567)}
\pgflineto{\pgfxy(283.50390814131504,105.22103837068119)}
\pgfclosepath 
\pgffill 
\pgfmoveto{\pgfxy(323,149)} 
\pgfcurveto{\pgfxy(338,142)}{\pgfxy(338,126)}{\pgfxy(323,120)}
\pgfstroke
\pgfmoveto{\pgfxy(323.0,149.0)}
\pgflineto{\pgfxy(325.7789616293188,145.496091858685)}
\pgflineto{\pgfxy(327.4705034906433,149.12082441866605)}
\pgfclosepath 
\pgffill 
\pgfline{\pgfxy(68.0,203.0)}{\pgfxy(68.0,176.0)}
\pgfline{\pgfxy(68.0,176.0)}{\pgfxy(77.0,160.0)}
\pgfline{\pgfxy(68.0,160.0)}{\pgfxy(68.0,135.0)}
\pgfmoveto{\pgfxy(55,176)} 
\pgfcurveto{\pgfxy(62,160)}{\pgfxy(77,160)}{\pgfxy(84,176)}
\pgfstroke
\pgfmoveto{\pgfxy(55.0,176.0)}
\pgflineto{\pgfxy(54.770960666274455,171.53373299235184)}
\pgflineto{\pgfxy(58.435590005883206,173.13700832843065)}
\pgfclosepath 
\pgffill 
\begin{pgfmagnify}{1}{-1}
\pgfputat{\pgfxy(116,-108)}{\pgfbox[left,top]{$e_1$}}
\end{pgfmagnify}
\begin{pgfmagnify}{1}{-1}
\pgfputat{\pgfxy(86,-126)}{\pgfbox[left,top]{$e_2$}}
\end{pgfmagnify}
\begin{pgfmagnify}{1}{-1}
\pgfputat{\pgfxy(129,-126)}{\pgfbox[left,top]{$e_3$}}
\end{pgfmagnify}
\begin{pgfmagnify}{1}{-1}
\pgfputat{\pgfxy(161,-66)}{\pgfbox[left,top]{$e_4$}}
\end{pgfmagnify}
\begin{pgfmagnify}{1}{-1}
\pgfputat{\pgfxy(161,-111)}{\pgfbox[left,top]{$e_5$}}
\end{pgfmagnify}
\begin{pgfmagnify}{1}{-1}
\pgfputat{\pgfxy(174,-126)}{\pgfbox[left,top]{$e_6$}}
\end{pgfmagnify}
\begin{pgfmagnify}{1}{-1}
\pgfputat{\pgfxy(161,-165)}{\pgfbox[left,top]{$e_7$}}
\end{pgfmagnify}
\begin{pgfmagnify}{1}{-1}
\pgfputat{\pgfxy(129,-194)}{\pgfbox[left,top]{$e_8$}}
\end{pgfmagnify}
\begin{pgfmagnify}{1}{-1}
\pgfputat{\pgfxy(86,-194)}{\pgfbox[left,top]{$e_9$}}
\end{pgfmagnify}
\begin{pgfmagnify}{1}{-1}
\pgfputat{\pgfxy(35,-194)}{\pgfbox[left,top]{$e_{10}$}}
\end{pgfmagnify}
\begin{pgfmagnify}{1}{-1}
\pgfputat{\pgfxy(221,-126)}{\pgfbox[left,top]{$e_{11}$}}
\end{pgfmagnify}
\begin{pgfmagnify}{1}{-1}
\pgfputat{\pgfxy(251,-156)}{\pgfbox[left,top]{$e_{12}$}}
\end{pgfmagnify}
\begin{pgfmagnify}{1}{-1}
\pgfputat{\pgfxy(269,-126)}{\pgfbox[left,top]{$e_{13}$}}
\end{pgfmagnify}
\begin{pgfmagnify}{1}{-1}
\pgfputat{\pgfxy(296,-156)}{\pgfbox[left,top]{$e_{14}$}}
\end{pgfmagnify}
\begin{pgfmagnify}{1}{-1}
\pgfputat{\pgfxy(296,-79)}{\pgfbox[left,top]{$e_{15}$}}
\end{pgfmagnify}
\begin{pgfmagnify}{1}{-1}
\pgfputat{\pgfxy(311,-99)}{\pgfbox[left,top]{$SW_1$}}
\end{pgfmagnify}
\begin{pgfmagnify}{1}{-1}
\pgfputat{\pgfxy(48,-147)}{\pgfbox[left,top]{$e_{18}$}}
\end{pgfmagnify}
\begin{pgfmagnify}{1}{-1}
\pgfputat{\pgfxy(89,-169)}{\pgfbox[left,top]{$SW_2$}}
\end{pgfmagnify}
\begin{pgfmagnify}{1}{-1}
\pgfputat{\pgfxy(239,-97)}{\pgfbox[left,top]{$SW_3$}}
\end{pgfmagnify}
\begin{pgfmagnify}{1}{-1}
\pgfputat{\pgfxy(327,-156)}{\pgfbox[left,top]{$SW_4$}}
\end{pgfmagnify}
\begin{pgfmagnify}{1}{-1}
\pgfputat{\pgfxy(343,-126)}{\pgfbox[left,top]{$e_{20}$}}
\end{pgfmagnify}
\begin{pgfmagnify}{1}{-1}
\pgfputat{\pgfxy(374,-108)}{\pgfbox[left,top]{$e_{22}$}}
\end{pgfmagnify}
\begin{pgfmagnify}{1}{-1}
\pgfputat{\pgfxy(390,-126)}{\pgfbox[left,top]{$e_{21}$}}
\end{pgfmagnify}
\begin{pgfmagnify}{1}{-1}
\pgfputat{\pgfxy(262,-59)}{\pgfbox[left,top]{$e_{16}$}}
\end{pgfmagnify}
\begin{pgfmagnify}{1}{-1}
\pgfputat{\pgfxy(296,-41)}{\pgfbox[left,top]{$e_{17}$}}
\end{pgfmagnify}
\begin{pgfmagnify}{1}{-1}
\pgfputat{\pgfxy(221,-77)}{\pgfbox[left,top]{$e_{19}$}}
\end{pgfmagnify}
\color{layer2}
\pgfmoveto{\pgfxy(154,131)}
\pgflineto{\pgfxy(163,131)}
\pgflineto{\pgfxy(163,140)}
\pgflineto{\pgfxy(154,140)}
\pgfclosepath 
\pgffill 
\pgfmoveto{\pgfxy(289,176)}
\pgflineto{\pgfxy(298,176)}
\pgflineto{\pgfxy(298,185)}
\pgflineto{\pgfxy(289,185)}
\pgfclosepath 
\pgffill 
\begin{pgfmagnify}{1}{-1}
\pgfputat{\pgfxy(165,-207)}{\pgfbox[left,top]{$4$}}
\end{pgfmagnify}
\begin{pgfmagnify}{1}{-1}
\pgfputat{\pgfxy(156,-30)}{\pgfbox[left,top]{$2$}}
\end{pgfmagnify}
\pgfmoveto{\pgfxy(154,41)}
\pgflineto{\pgfxy(163,41)}
\pgflineto{\pgfxy(163,50)}
\pgflineto{\pgfxy(154,50)}
\pgfclosepath 
\pgffill 
\begin{pgfmagnify}{1}{-1}
\pgfputat{\pgfxy(111,-75)}{\pgfbox[left,top]{$1$}}
\end{pgfmagnify}
\pgfmoveto{\pgfxy(368,86)}
\pgflineto{\pgfxy(377,86)}
\pgflineto{\pgfxy(377,95)}
\pgflineto{\pgfxy(368,95)}
\pgfclosepath 
\pgffill 
\begin{pgfmagnify}{1}{-1}
\pgfputat{\pgfxy(167,-142)}{\pgfbox[left,top]{$3$}}
\end{pgfmagnify}
\pgfmoveto{\pgfxy(109,86)}
\pgflineto{\pgfxy(118,86)}
\pgflineto{\pgfxy(118,95)}
\pgflineto{\pgfxy(109,95)}
\pgfclosepath 
\pgffill 
\pgfmoveto{\pgfxy(154,198)}
\pgflineto{\pgfxy(163,198)}
\pgflineto{\pgfxy(163,207)}
\pgflineto{\pgfxy(154,207)}
\pgfclosepath 
\pgffill 
\begin{pgfmagnify}{1}{-1}
\pgfputat{\pgfxy(291,-189)}{\pgfbox[left,top]{$5$}}
\end{pgfmagnify}
\begin{pgfmagnify}{1}{-1}
\pgfputat{\pgfxy(368,-75)}{\pgfbox[left,top]{$21$}}
\end{pgfmagnify}
\color{layer11}
\pgfellipse[fillstroke]{\pgfxy(158.5,90.5)}{\pgfxy(4.5,0)}{\pgfxy(0,4.5)}
\pgfellipse[fillstroke]{\pgfxy(293.5,135.5)}{\pgfxy(4.5,0)}{\pgfxy(0,4.5)}
\pgfellipse[fillstroke]{\pgfxy(293.5,22.5)}{\pgfxy(4.5,0)}{\pgfxy(0,4.5)}
\pgfellipse[fillstroke]{\pgfxy(372.5,135.5)}{\pgfxy(4.5,0)}{\pgfxy(0,4.5)}
\begin{pgfmagnify}{1}{-1}
\pgfputat{\pgfxy(64,-122)}{\pgfbox[left,top]{$6$}}
\end{pgfmagnify}
\pgfellipse[fillstroke]{\pgfxy(68.5,135.5)}{\pgfxy(4.5,0)}{\pgfxy(0,4.5)}
\pgfellipse[fillstroke]{\pgfxy(248.5,67.5)}{\pgfxy(4.5,0)}{\pgfxy(0,4.5)}
\begin{pgfmagnify}{1}{-1}
\pgfputat{\pgfxy(111,-144)}{\pgfbox[left,top]{$7$}}
\end{pgfmagnify}
\pgfellipse[fillstroke]{\pgfxy(113.5,202.5)}{\pgfxy(4.5,0)}{\pgfxy(0,4.5)}
\begin{pgfmagnify}{1}{-1}
\pgfputat{\pgfxy(165,-86)}{\pgfbox[left,top]{$8$}}
\end{pgfmagnify}
\begin{pgfmagnify}{1}{-1}
\pgfputat{\pgfxy(8,-212)}{\pgfbox[left,top]{$9$}}
\end{pgfmagnify}
\begin{pgfmagnify}{1}{-1}
\pgfputat{\pgfxy(62,-212)}{\pgfbox[left,top]{$10$}}
\end{pgfmagnify}
\pgfellipse[fillstroke]{\pgfxy(203.5,135.5)}{\pgfxy(4.5,0)}{\pgfxy(0,4.5)}
\begin{pgfmagnify}{1}{-1}
\pgfputat{\pgfxy(109,-212)}{\pgfbox[left,top]{$11$}}
\end{pgfmagnify}
\begin{pgfmagnify}{1}{-1}
\pgfputat{\pgfxy(199,-144)}{\pgfbox[left,top]{$12$}}
\end{pgfmagnify}
\begin{pgfmagnify}{1}{-1}
\pgfputat{\pgfxy(244,-122)}{\pgfbox[left,top]{$13$}}
\end{pgfmagnify}
\pgfellipse[fillstroke]{\pgfxy(248.5,135.5)}{\pgfxy(4.5,0)}{\pgfxy(0,4.5)}
\begin{pgfmagnify}{1}{-1}
\pgfputat{\pgfxy(244,-189)}{\pgfbox[left,top]{$14$}}
\end{pgfmagnify}
\begin{pgfmagnify}{1}{-1}
\pgfputat{\pgfxy(242,-52)}{\pgfbox[left,top]{$16$}}
\end{pgfmagnify}
\begin{pgfmagnify}{1}{-1}
\pgfputat{\pgfxy(289,-7)}{\pgfbox[left,top]{$17$}}
\end{pgfmagnify}
\pgfellipse[fillstroke]{\pgfxy(12.5,202.5)}{\pgfxy(4.5,0)}{\pgfxy(0,4.5)}
\begin{pgfmagnify}{1}{-1}
\pgfputat{\pgfxy(300,-61)}{\pgfbox[left,top]{$18$}}
\end{pgfmagnify}
\begin{pgfmagnify}{1}{-1}
\pgfputat{\pgfxy(298,-140)}{\pgfbox[left,top]{$15$}}
\end{pgfmagnify}
\begin{pgfmagnify}{1}{-1}
\pgfputat{\pgfxy(368,-144)}{\pgfbox[left,top]{$19$}}
\end{pgfmagnify}
\pgfellipse[fillstroke]{\pgfxy(68.5,202.5)}{\pgfxy(4.5,0)}{\pgfxy(0,4.5)}
\begin{pgfmagnify}{1}{-1}
\pgfputat{\pgfxy(426,-133)}{\pgfbox[left,top]{$20$}}
\end{pgfmagnify}
\pgfellipse[fillstroke]{\pgfxy(113.5,135.5)}{\pgfxy(4.5,0)}{\pgfxy(0,4.5)}
\pgfellipse[fillstroke]{\pgfxy(248.5,180.5)}{\pgfxy(4.5,0)}{\pgfxy(0,4.5)}
\pgfellipse[fillstroke]{\pgfxy(293.5,67.5)}{\pgfxy(4.5,0)}{\pgfxy(0,4.5)}
\pgfellipse[fillstroke]{\pgfxy(417.5,135.5)}{\pgfxy(4.5,0)}{\pgfxy(0,4.5)}
\end{pgfmagnify}
\end{pgfpicture}
		\end{small}
		\caption{The proposed 21-bus network. Red squares denote boundary nodes (DGUs), blue circles represent interior nodes (loads).}
		\label{Fig6:Retemia} 
	\end{figure}
	
	 The lines parameters of the original network are collected in Table \ref{Appendix:Table21original}, in Appendix \ref{Appendix:label}. These parameters have been chosen so that the time constants are spread in a wide range; this fact let us show that PnP control can stabilize networks made of lines with very different characteristics. As regards the loads connected to the buses, they are listed in Tables \ref{Appendix:Table21Linloads} and \ref{Appendix:Table21loads} in Appendix \ref{Appendix:label}.

At time $t=0$ s, there is no energy stored in all inductors and capacitors and all the switches are open. At instant $t = 5$ s, switch $SW_1$ closes: the overall effect is an increase in electrical loads, mainly supplied by DGUs 3 and 5. Then, $SW_2$ e $SW_3$ close respectively at instants $t = 6.5$ s and $t = 8$ s, connecting new branches to the network. Finally, at $t = 9.5$ s switch $SW_4$ closes, so that the sixth generation unit is connected to the microgrid. Table \ref{chap6:Tabletimeevents} summarizes the commutations sequence. 
		\begin{table}[h]
		\centering
		\begin{tabular}{cc|cccc}
			\toprule
			Time [s] & Event & $SW_1$ & $SW_2$ & $SW_3$ & $SW_4$ \\
			\midrule
			0 & Start simulation & open & open & open & open \\
			5 & $SW_1$ closes & closed & open & open & open \\
			6.5 & $SW_2$ closes & closed & closed & open & open \\
			8 & $SW_3$ closes & closed & closed & closed & open \\
			9.5 & $SW_4$ closes & closed & closed & closed & closed \\
			13 & End simulation & closed & closed & closed & closed \\
			\bottomrule
		\end{tabular}
		\caption{Simulation time events.}
		\label{chap6:Tabletimeevents}
	\end{table}
		
		
	
\subsection{PnP control design}
As described in Section \ref{sec:PnPdesfen}, the first step in control design consists in applying hKR for obtaining an equivalent ImG with a load-connected topology. Equivalent line impedances are obtained as in \eqref{eq:star} with $\omega_0=2\pi50$ rad/s.
The reduced network actually depends on the state of the switches in
the original network. As long as $SW_2$, $SW_3$ and $SW_4$ are open,
the network in Figure \ref{Fig6:Retemia} has a radial topology. In
particular, the equivalent impedance between nodes 3 and 5 is equal to
the sum of the impedances of edges $e_{6}$, $e_{11}$, $e_{13}$ and
$e_{14}$, irrespectively of the state of switch $SW_1$. On the contrary,
when switches $SW_2$, $SW_3$ and $SW_4$ are closed, the topology of
the hybrid Kron reduced network and its impedances change. Figure
\ref{fig:21nodes_KR} collects the reduced networks that arise
during the simulation. 
\begin{figure}
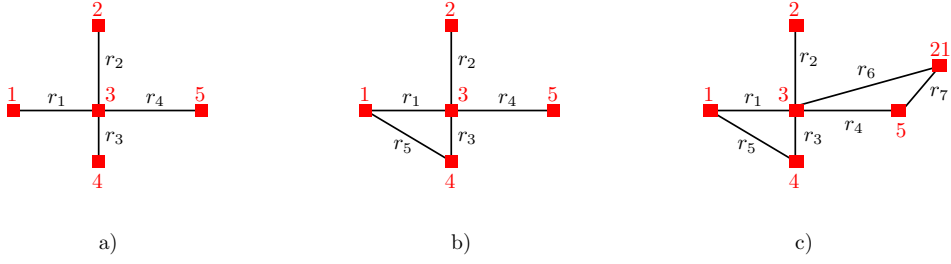

		\centering
		\scalebox{0.8}{\begin{pgfpicture}{0cm}{0cm}{465pt}{126pt}
\pgfsetxvec{\pgfpoint{1pt}{0pt}}
\pgfsetyvec{\pgfpoint{0pt}{1pt}}
\pgfsetroundjoin 
\pgfsetroundcap
\pgftranslateto{\pgfxy(0,126)}
\begin{pgfmagnify}{1}{-1}
\definecolor{layer0}{rgb}{0.0,0.0,0.0}
\definecolor{layer1}{rgb}{0.0,0.0,0.5}
\definecolor{layer2}{rgb}{1.0,0.0,0.0}
\definecolor{layer3}{rgb}{0.0,0.5,0.5}
\definecolor{layer4}{rgb}{1.0,0.78,0.0}
\definecolor{layer5}{rgb}{0.5,1.0,0.0}
\definecolor{layer6}{rgb}{0.0,1.0,1.0}
\definecolor{layer7}{rgb}{0.0,0.5,0.0}
\definecolor{layer8}{rgb}{0.6,0.8,0.2}
\definecolor{layer9}{rgb}{1.0,0.08,0.58}
\definecolor{layer10}{rgb}{0.71,0.61,0.05}
\definecolor{layer11}{rgb}{0.0,0.5,1.0}
\definecolor{layer12}{rgb}{0.88,0.88,0.88}
\definecolor{layer13}{rgb}{0.64,0.64,0.64}
\definecolor{layer14}{rgb}{0.37,0.37,0.37}
\definecolor{layer15}{rgb}{0.0,0.0,0.0}
\color{layer0}
\pgfsetlinewidth{0.8pt}
\pgfline{\pgfxy(9.0,57.0)}{\pgfxy(97.0,57.0)}
\pgfline{\pgfxy(49.0,17.0)}{\pgfxy(49.0,81.0)}
\begin{pgfmagnify}{1}{-1}
\pgfputat{\pgfxy(52,-31)}{\pgfbox[left,top]{$r_2$}}
\end{pgfmagnify}
\begin{pgfmagnify}{1}{-1}
\pgfputat{\pgfxy(52,-67)}{\pgfbox[left,top]{$r_3$}}
\end{pgfmagnify}
\begin{pgfmagnify}{1}{-1}
\pgfputat{\pgfxy(71,-49)}{\pgfbox[left,top]{$r_4$}}
\end{pgfmagnify}
\begin{pgfmagnify}{1}{-1}
\pgfputat{\pgfxy(25,-49)}{\pgfbox[left,top]{$r_1$}}
\end{pgfmagnify}
\begin{pgfmagnify}{1}{-1}
\pgfputat{\pgfxy(49,-115)}{\pgfbox[left,top]{a)}}
\end{pgfmagnify}
\begin{pgfmagnify}{1}{-1}
\pgfputat{\pgfxy(214,-115)}{\pgfbox[left,top]{b)}}
\end{pgfmagnify}
\begin{pgfmagnify}{1}{-1}
\pgfputat{\pgfxy(217,-67)}{\pgfbox[left,top]{$r_3$}}
\end{pgfmagnify}
\pgfline{\pgfxy(174.0,57.0)}{\pgfxy(262.0,57.0)}
\pgfline{\pgfxy(214.0,17.0)}{\pgfxy(214.0,81.0)}
\begin{pgfmagnify}{1}{-1}
\pgfputat{\pgfxy(217,-31)}{\pgfbox[left,top]{$r_2$}}
\end{pgfmagnify}
\begin{pgfmagnify}{1}{-1}
\pgfputat{\pgfxy(191,-49)}{\pgfbox[left,top]{$r_1$}}
\end{pgfmagnify}
\begin{pgfmagnify}{1}{-1}
\pgfputat{\pgfxy(236,-49)}{\pgfbox[left,top]{$r_4$}}
\end{pgfmagnify}
\pgfline{\pgfxy(174.0,57.0)}{\pgfxy(214.0,81.0)}
\begin{pgfmagnify}{1}{-1}
\pgfputat{\pgfxy(187,-71)}{\pgfbox[left,top]{$r_5$}}
\end{pgfmagnify}
\pgfline{\pgfxy(335.0,57.0)}{\pgfxy(375.0,81.0)}
\begin{pgfmagnify}{1}{-1}
\pgfputat{\pgfxy(348,-71)}{\pgfbox[left,top]{$r_5$}}
\end{pgfmagnify}
\begin{pgfmagnify}{1}{-1}
\pgfputat{\pgfxy(379,-67)}{\pgfbox[left,top]{$r_3$}}
\end{pgfmagnify}
\pgfline{\pgfxy(335.0,57.0)}{\pgfxy(423.0,57.0)}
\pgfline{\pgfxy(375.0,17.0)}{\pgfxy(375.0,81.0)}
\begin{pgfmagnify}{1}{-1}
\pgfputat{\pgfxy(377,-30)}{\pgfbox[left,top]{$r_2$}}
\end{pgfmagnify}
\begin{pgfmagnify}{1}{-1}
\pgfputat{\pgfxy(351,-49)}{\pgfbox[left,top]{$r_1$}}
\end{pgfmagnify}
\begin{pgfmagnify}{1}{-1}
\pgfputat{\pgfxy(375,-115)}{\pgfbox[left,top]{c)}}
\end{pgfmagnify}
\begin{pgfmagnify}{1}{-1}
\pgfputat{\pgfxy(398,-62)}{\pgfbox[left,top]{$r_4$}}
\end{pgfmagnify}
\pgfline{\pgfxy(375.0,55.0)}{\pgfxy(443.0,36.0)}
\pgfline{\pgfxy(425.0,57.0)}{\pgfxy(441.0,38.0)}
\begin{pgfmagnify}{1}{-1}
\pgfputat{\pgfxy(438,-46)}{\pgfbox[left,top]{$r_7$}}
\end{pgfmagnify}
\begin{pgfmagnify}{1}{-1}
\pgfputat{\pgfxy(404,-36)}{\pgfbox[left,top]{$r_6$}}
\end{pgfmagnify}
\color{layer2}
\pgfmoveto{\pgfxy(46,78)}
\pgflineto{\pgfxy(52,78)}
\pgflineto{\pgfxy(52,84)}
\pgflineto{\pgfxy(46,84)}
\pgfclosepath 
\pgffill 
\begin{pgfmagnify}{1}{-1}
\pgfputat{\pgfxy(46,-6)}{\pgfbox[left,top]{$2$}}
\end{pgfmagnify}
\begin{pgfmagnify}{1}{-1}
\pgfputat{\pgfxy(217,-46)}{\pgfbox[left,top]{$3$}}
\end{pgfmagnify}
\pgfmoveto{\pgfxy(259,54)}
\pgflineto{\pgfxy(265,54)}
\pgflineto{\pgfxy(265,60)}
\pgflineto{\pgfxy(259,60)}
\pgfclosepath 
\pgffill 
\begin{pgfmagnify}{1}{-1}
\pgfputat{\pgfxy(211,-87)}{\pgfbox[left,top]{$4$}}
\end{pgfmagnify}
\begin{pgfmagnify}{1}{-1}
\pgfputat{\pgfxy(94,-46)}{\pgfbox[left,top]{$5$}}
\end{pgfmagnify}
\pgfmoveto{\pgfxy(46,54)}
\pgflineto{\pgfxy(52,54)}
\pgflineto{\pgfxy(52,60)}
\pgflineto{\pgfxy(46,60)}
\pgfclosepath 
\pgffill 
\pgfmoveto{\pgfxy(94,54)}
\pgflineto{\pgfxy(100,54)}
\pgflineto{\pgfxy(100,60)}
\pgflineto{\pgfxy(94,60)}
\pgfclosepath 
\pgffill 
\pgfmoveto{\pgfxy(211,54)}
\pgflineto{\pgfxy(217,54)}
\pgflineto{\pgfxy(217,60)}
\pgflineto{\pgfxy(211,60)}
\pgfclosepath 
\pgffill 
\pgfmoveto{\pgfxy(211,14)}
\pgflineto{\pgfxy(217,14)}
\pgflineto{\pgfxy(217,20)}
\pgflineto{\pgfxy(211,20)}
\pgfclosepath 
\pgffill 
\begin{pgfmagnify}{1}{-1}
\pgfputat{\pgfxy(259,-46)}{\pgfbox[left,top]{$5$}}
\end{pgfmagnify}
\begin{pgfmagnify}{1}{-1}
\pgfputat{\pgfxy(211,-6)}{\pgfbox[left,top]{$2$}}
\end{pgfmagnify}
\pgfmoveto{\pgfxy(171,54)}
\pgflineto{\pgfxy(177,54)}
\pgflineto{\pgfxy(177,60)}
\pgflineto{\pgfxy(171,60)}
\pgfclosepath 
\pgffill 
\begin{pgfmagnify}{1}{-1}
\pgfputat{\pgfxy(171,-46)}{\pgfbox[left,top]{$1$}}
\end{pgfmagnify}
\begin{pgfmagnify}{1}{-1}
\pgfputat{\pgfxy(46,-87)}{\pgfbox[left,top]{$4$}}
\end{pgfmagnify}
\pgfmoveto{\pgfxy(46,14)}
\pgflineto{\pgfxy(52,14)}
\pgflineto{\pgfxy(52,20)}
\pgflineto{\pgfxy(46,20)}
\pgfclosepath 
\pgffill 
\pgfmoveto{\pgfxy(372,54)}
\pgflineto{\pgfxy(379,54)}
\pgflineto{\pgfxy(379,60)}
\pgflineto{\pgfxy(372,60)}
\pgfclosepath 
\pgffill 
\pgfmoveto{\pgfxy(372,14)}
\pgflineto{\pgfxy(379,14)}
\pgflineto{\pgfxy(379,20)}
\pgflineto{\pgfxy(372,20)}
\pgfclosepath 
\pgffill 
\begin{pgfmagnify}{1}{-1}
\pgfputat{\pgfxy(422,-63)}{\pgfbox[left,top]{$5$}}
\end{pgfmagnify}
\begin{pgfmagnify}{1}{-1}
\pgfputat{\pgfxy(372,-6)}{\pgfbox[left,top]{$2$}}
\end{pgfmagnify}
\pgfmoveto{\pgfxy(332,54)}
\pgflineto{\pgfxy(339,54)}
\pgflineto{\pgfxy(339,60)}
\pgflineto{\pgfxy(332,60)}
\pgfclosepath 
\pgffill 
\begin{pgfmagnify}{1}{-1}
\pgfputat{\pgfxy(332,-46)}{\pgfbox[left,top]{$1$}}
\end{pgfmagnify}
\pgfmoveto{\pgfxy(372,78)}
\pgflineto{\pgfxy(379,78)}
\pgflineto{\pgfxy(379,84)}
\pgflineto{\pgfxy(372,84)}
\pgfclosepath 
\pgffill 
\begin{pgfmagnify}{1}{-1}
\pgfputat{\pgfxy(367,-46)}{\pgfbox[left,top]{$3$}}
\end{pgfmagnify}
\pgfmoveto{\pgfxy(420,54)}
\pgflineto{\pgfxy(427,54)}
\pgflineto{\pgfxy(427,60)}
\pgflineto{\pgfxy(420,60)}
\pgfclosepath 
\pgffill 
\begin{pgfmagnify}{1}{-1}
\pgfputat{\pgfxy(372,-87)}{\pgfbox[left,top]{$4$}}
\end{pgfmagnify}
\begin{pgfmagnify}{1}{-1}
\pgfputat{\pgfxy(438,-25)}{\pgfbox[left,top]{$21$}}
\end{pgfmagnify}
\pgfmoveto{\pgfxy(439,33)}
\pgflineto{\pgfxy(446,33)}
\pgflineto{\pgfxy(446,39)}
\pgflineto{\pgfxy(439,39)}
\pgfclosepath 
\pgffill 
\begin{pgfmagnify}{1}{-1}
\pgfputat{\pgfxy(52,-46)}{\pgfbox[left,top]{$3$}}
\end{pgfmagnify}
\begin{pgfmagnify}{1}{-1}
\pgfputat{\pgfxy(6,-46)}{\pgfbox[left,top]{$1$}}
\end{pgfmagnify}
\pgfmoveto{\pgfxy(6,54)}
\pgflineto{\pgfxy(12,54)}
\pgflineto{\pgfxy(12,60)}
\pgflineto{\pgfxy(6,60)}
\pgfclosepath 
\pgffill 
\pgfmoveto{\pgfxy(211,78)}
\pgflineto{\pgfxy(217,78)}
\pgflineto{\pgfxy(217,84)}
\pgflineto{\pgfxy(211,84)}
\pgfclosepath 
\pgffill 
\end{pgfmagnify}
\end{pgfpicture}}
\caption{Hybrid Kron reduced networks. Figures a)-b) and c) show the
  equivalent network for $t\leq 6.5$ s, $6.5\leq t\leq 9.5$ s and $t\geq 9.5$ s, respectively.}
		\label{fig:21nodes_KR}
	\end{figure}
In particular, Figure \ref{fig:21nodes_KR}a
holds when all the four switches are open, or when only $SW_1$ is
closed. The network in Figure \ref{fig:21nodes_KR}b refers to the
case when $SW_2$ and $SW_3$ become closed. The diagram in
Figure \ref{fig:21nodes_KR}c, instead, holds when all the four switches are closed.
The resistances and inductances of the Kron reduced circuits are
listed in Tables \ref{Appendix:Table21EqParam1}, \ref{Appendix:Table21EqParam2}, \ref{Appendix:Table21EqParam3} and \ref{Appendix:Table21EqParam4} in Appendix \ref{Appendix:label}. Notably, we remark that, for the three reduced networks, all the equivalent resistances and inductances are positive.
At instant $t = 5$ s, $SW_1$ closes, but the hybrid Kron reduced network does not
change, and no redesign of the controllers is needed. Indeed, the only effect of load currents in 16, 17, 18 is to change the term $\tilde I_b(s)$ in \eqref{hKRa}, which is lumped, anyway, in the disturbance loads attached to the PCCs of DGUs. At instant $t =
6.5 \,\mbox{s}$, the equivalent impedances between nodes 1, 3 and 4 change: the
controllers of DGUs 1, 3 and 4 must be redesigned (all the other
controllers do not change). At instant $t = 8$ s, the equivalent
impedance between boundary nodes 3 and 5 changes: the corresponding
DGUs must update their controllers. At $t = 9.5$ s, DGU 21 is connected
to the microgrid: controllers of DGUs 3 e 5 must be redesigned
again.\\
	The proposed controllers effectively stabilize the microgrid:
        this is illustrated by Figures \ref{Fig6:Matlab1} -\ref{Fig6:Matlab4}. These figures show, for every
        configuration of the switches, the singular values of the
        closed loop QSL microgrid, as well as its eigenvalues, with
        and without couplings. In particular, Figures \ref{fig:Matlab1a}, \ref{fig:Matlab2a}, \ref{fig:Matlab3a} and \ref{fig:Matlab4a} show that all the eigenvalues are on the left half-plane.

        \begin{figure}[!htb]
                      \centering
                      \begin{subfigure}[!htb]{0.48\textwidth}
                        \centering
                        \includegraphics[width=1\textwidth]{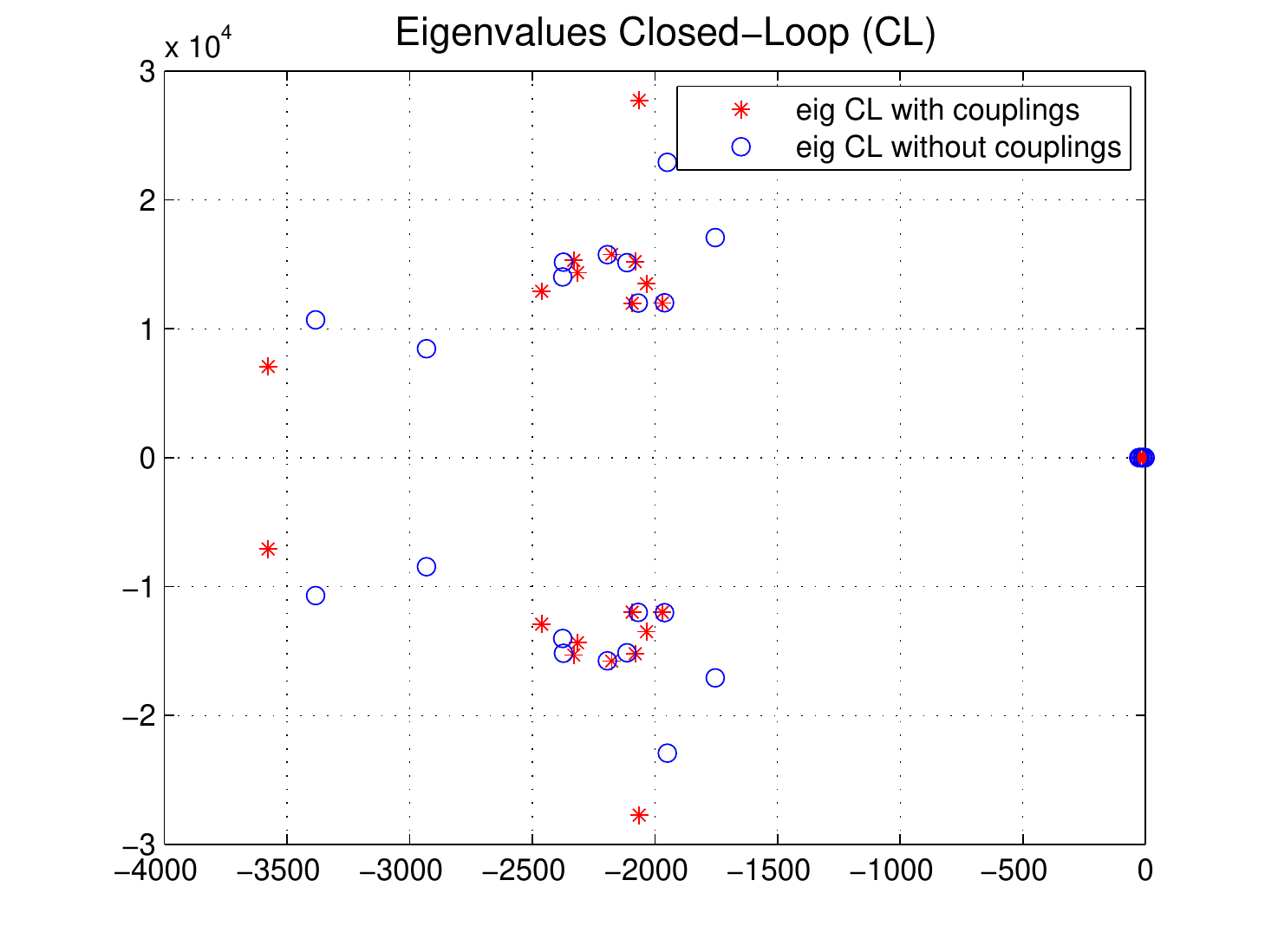}
                        \caption{Eigenvalues of the closed-loop (CL) QSL microgrid with (red) and without (blue) couplings.}
                        \label{fig:Matlab1a}
                      \end{subfigure}
                      \begin{subfigure}[!htb]{0.48\textwidth}
                        \centering
                        \includegraphics[width=1\textwidth]{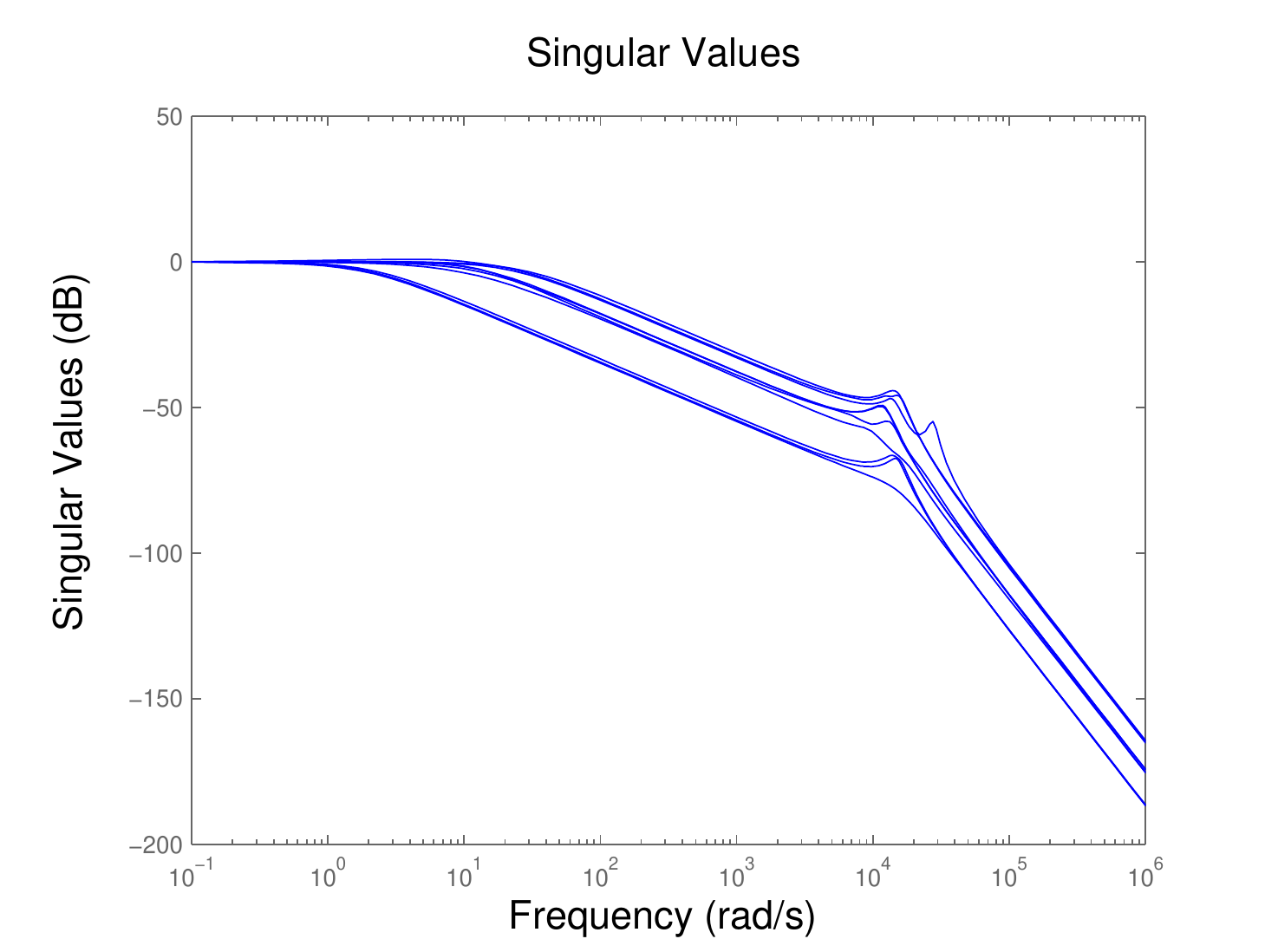}
                        \caption{Singular values of the closed-loop (CL) QSL microgrid.}
                        \label{fig:Matlab1b}
                      \end{subfigure}
		\caption{Eigenvalues and singular values of the QSL microgrid when all the switches are open, or only $SW_1$ is closed.}
		\label{Fig6:Matlab1}
                    \end{figure}

      \begin{figure}[!htb]
                      \centering
                      \begin{subfigure}[!htb]{0.48\textwidth}
                        \centering
                        \includegraphics[width=1\textwidth]{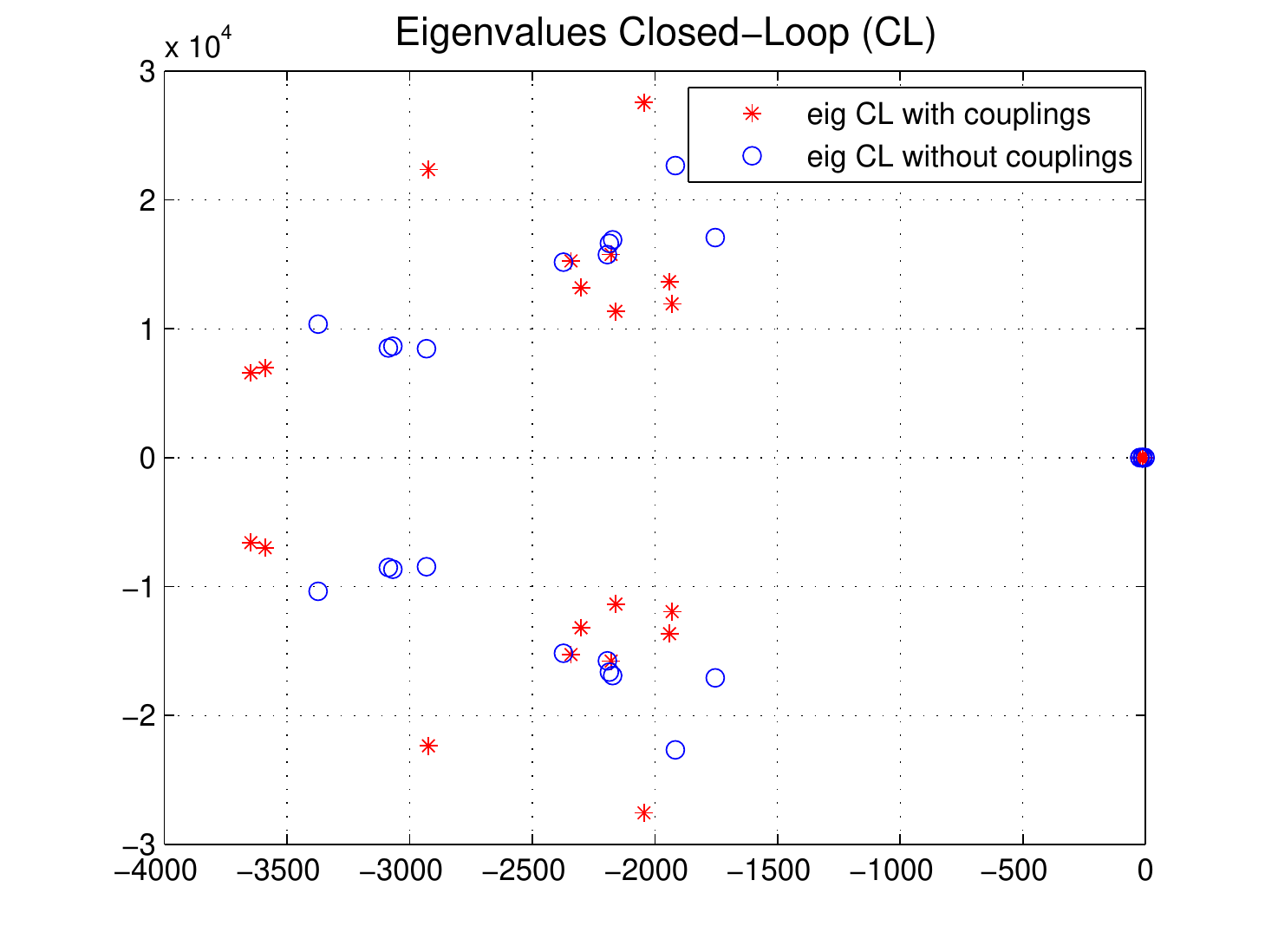}
                        \caption{Eigenvalues of the closed-loop (CL) QSL microgrid with (red) and without (blue) couplings.}
                        \label{fig:Matlab2a}
                      \end{subfigure}
                      \begin{subfigure}[!htb]{0.48\textwidth}
                        \centering
                        \includegraphics[width=1\textwidth]{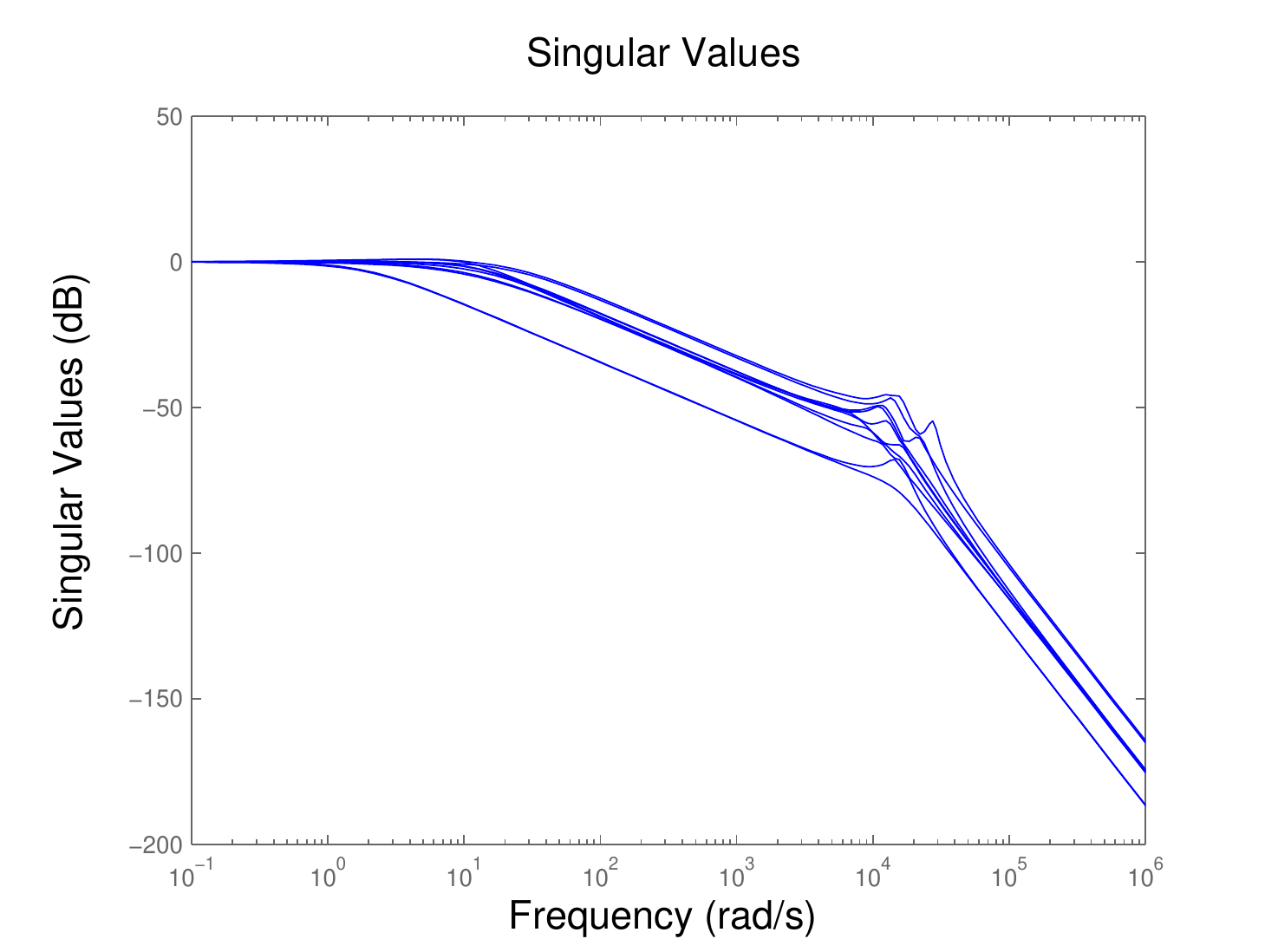}
                        \caption{Singular values of the closed-loop (CL) QSL microgrid.}
                        \label{fig:Matlab2b}
                      \end{subfigure}
		\caption{Eigenvalues and singular values with $SW_1$, $SW_2$ closed, and $SW_3$, $SW_4$ open.}
		\label{Fig6:Matlab2}
                    \end{figure}
 \begin{figure}[!htb]
                      \centering
       \begin{subfigure}[!htb]{0.48\textwidth}
                        \centering
                        \includegraphics[width=1\textwidth]{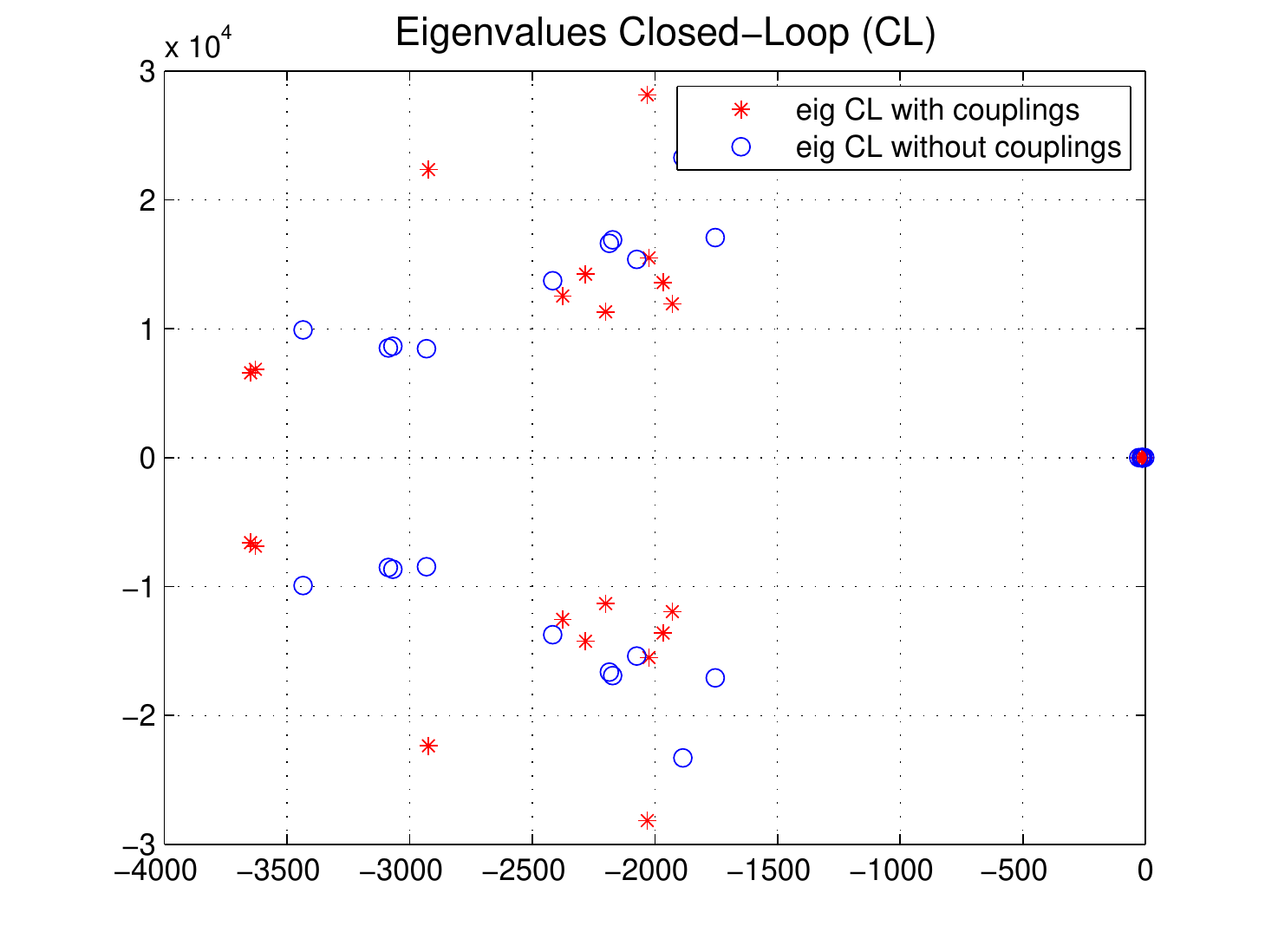}
                        \caption{Eigenvalues of the closed-loop (CL) QSL microgrid with (red) and without (blue) couplings.}
                        \label{fig:Matlab3a}
                      \end{subfigure}
                      \begin{subfigure}[!htb]{0.48\textwidth}
                        \centering
                        \includegraphics[width=1\textwidth]{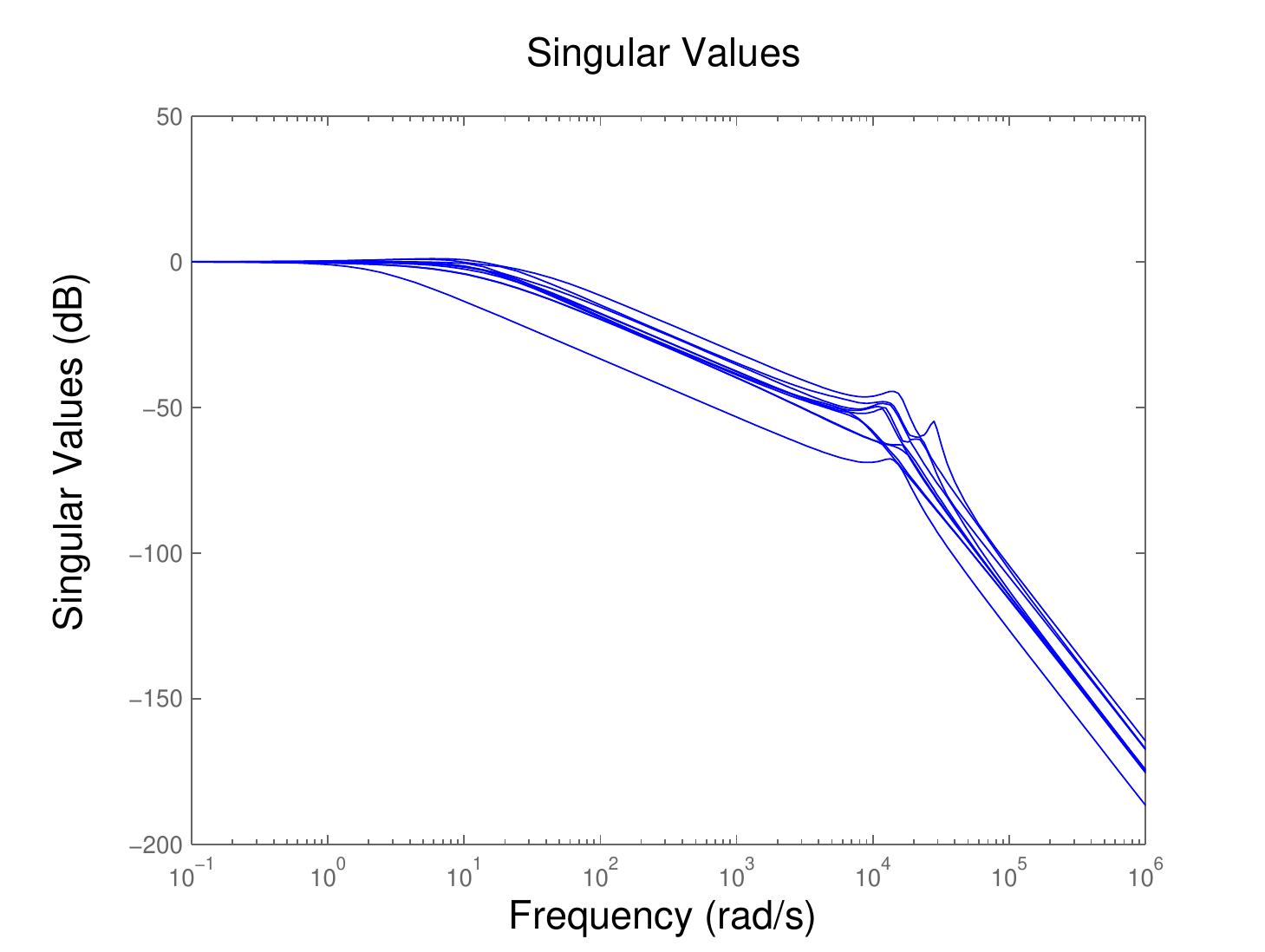}
                        \caption{Singular values of the closed-loop (CL) QSL microgrid.}
                        \label{fig:Matlab3b}
                      \end{subfigure}
		\caption{Eigenvalues and singular values when $SW_1$, $SW_2$, $SW_3$ are closed, and $SW_4$ open.}
		\label{Fig6:Matlab3}

\end{figure}
 \begin{figure}[!htb]
                      \centering
		     \begin{subfigure}[!htb]{0.48\textwidth}
                        \centering
                        \includegraphics[width=1\textwidth]{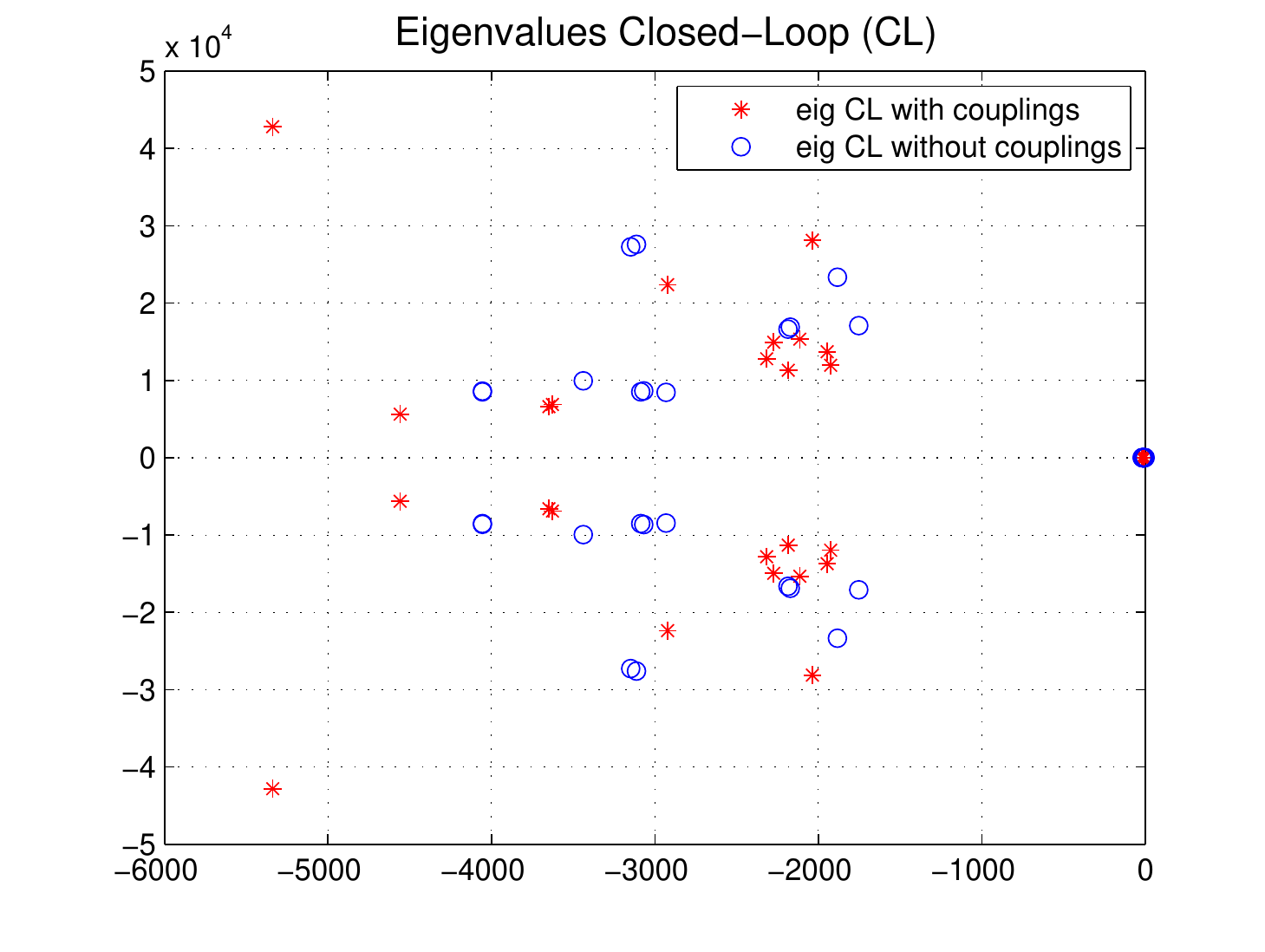}
                        \caption{Eigenvalues of the closed-loop (CL) QSL microgrid with (red) and without (blue) couplings.}
                        \label{fig:Matlab4a}
                      \end{subfigure}
                      \begin{subfigure}[!htb]{0.48\textwidth}
                        \centering
                        \includegraphics[width=1\textwidth]{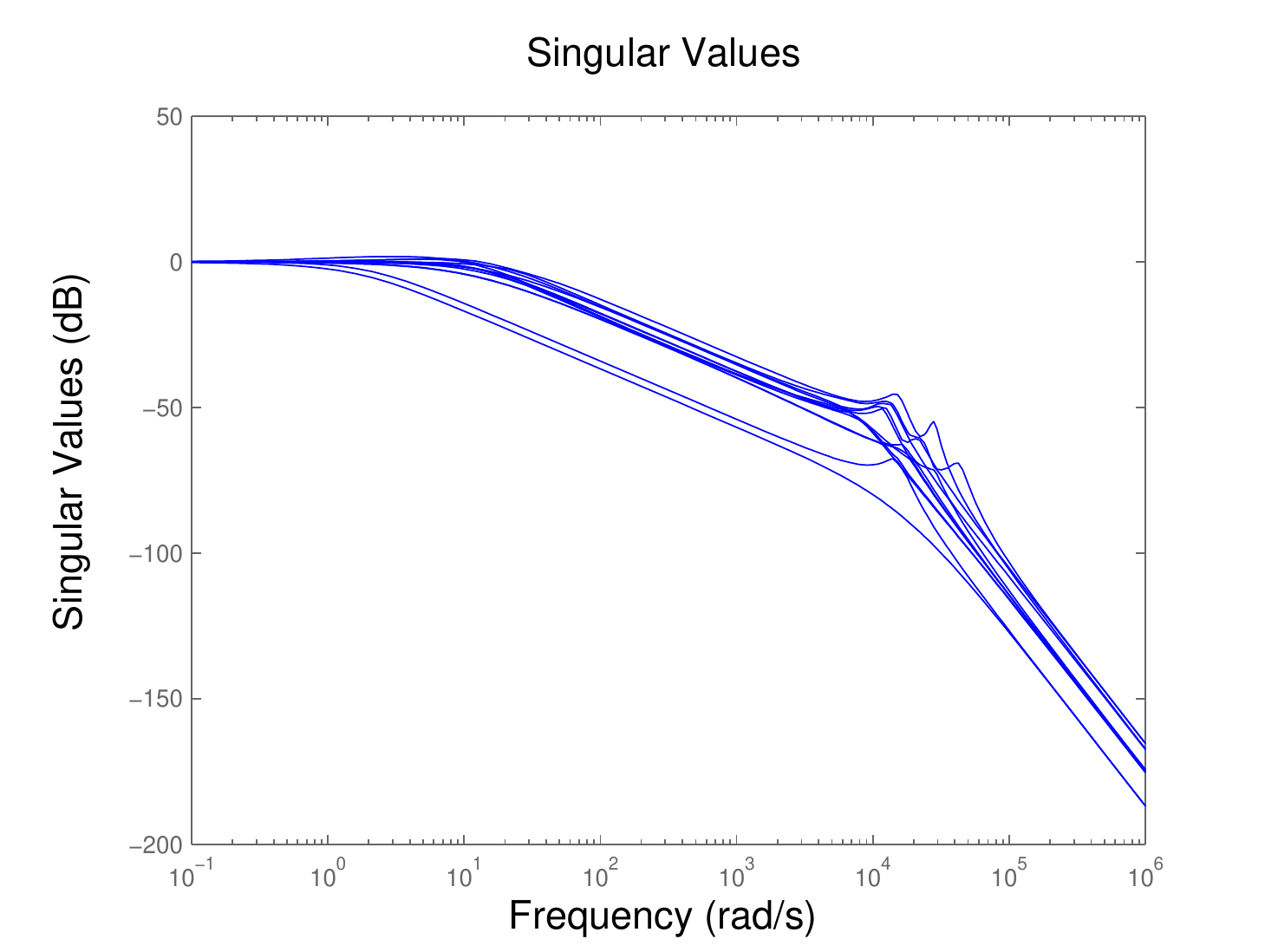}
                        \caption{Singular values of the closed-loop (CL) QSL microgrid.}
                        \label{fig:Matla4b}
                      \end{subfigure}
		\caption{Eigenvalues and singular values of the QSL microgrid when all the switches are closed.}
		\label{Fig6:Matlab4}

                    \end{figure}

\subsection{Simulation results}
    The reference signals for all the generation units are $V_d^{\mathrm{ref}} =
    \sqrt{2} \cdot 230\,\mbox{V}$ and $V_q^{\mathrm{ref}} = 0 \,\mbox{V}$. 
Figure \ref{fig:21nodes_performance} shows the Root Mean Square (RMS) voltage, frequency 
and Total Harmonic Distortion\footnote{See \cite{5154067} for a definition.} (THD) of phase $a$ at
the PCCs of the boundary nodes (i.e nodes 1, 2, 3, 4, 5 and 21), respectively. We highlight
that, in spite of all the variations of the ImG topology, PnP decentralized control ensures good tracking of
voltage references for all DGUs (see Figure
\ref{fig:21nodes_RMS}). We note that real-time switch between
different controllers has been implemented using a bumpless control
transfer scheme similar to the one used in classical PID
regulators \cite{aastrom2006advanced}. This guarantees control variables do not have sudden
variations at switching times. In our case, bumpless controllers are
effective in limiting voltage surges and dips to a
few volts when updates of controllers take place (for $t\geq 4$ s, the maximal deviation from the reference RMS voltage is of less than 20 volts). Figure \ref{fig:21nodes_f} shows
that the impact of the topology commutations is minor also on the frequency profiles. In fact, PnP controllers promptly restore the frequencies to the
nominal value, ensuring negligible variations (i.e. less than
$0.5\, \mathrm{Hz}$ when the highly inductive load is connected and less than $0.1\,\mathrm{Hz}$ when other events occur). Furthermore, from Figure \ref{fig:21nodes_THD}, we notice that THD values are
below the maximum limit ($5\%$) recommended in \cite{5154067}. \\
To evaluate the voltage imbalance at boundary nodes, we calculate the
ratio $V_N/V_P$ (expressed in $\%$), where $V_N$ and $V_P$ are the
magnitudes of the negative- and positive-sequence components of the
phase-to-phase voltage. The time evolution of this ratio is
represented in Figure \ref{fig:21nodes_Vimb}. We notice that it is
always below $0.5\%$ which is less than the maximum permissible value
$(3\%)$ defined by IEEE in \cite{5154067}. Finally, Figures \ref{fig:21nodes_ps} and \ref{fig:21nodes_qs} show respectively the active and
reactive power injected into the microgrid by the DGUs. We notice that
PnP controllers alone can not guarantee sharing of the reactive power
between the connected DGUs (see Figure \ref{fig:21nodes_qs}). This aim
can be achieved by means of higher level control schemes. However,
the fulfillment of any power flow requirements is outside the scope of
this work.

Overall, the fact that voltage and frequency stability is guaranteed
even for such a complicated network, proves that hKR is
a well suited method for extending the PnP scalable design to ImGs
with arbitrary topologies.
 \begin{figure}[htb]
                      \centering
                      \begin{subfigure}[htb]{0.48\textwidth}
                        \centering
                        \includegraphics[width=1.05\textwidth, height=150pt]{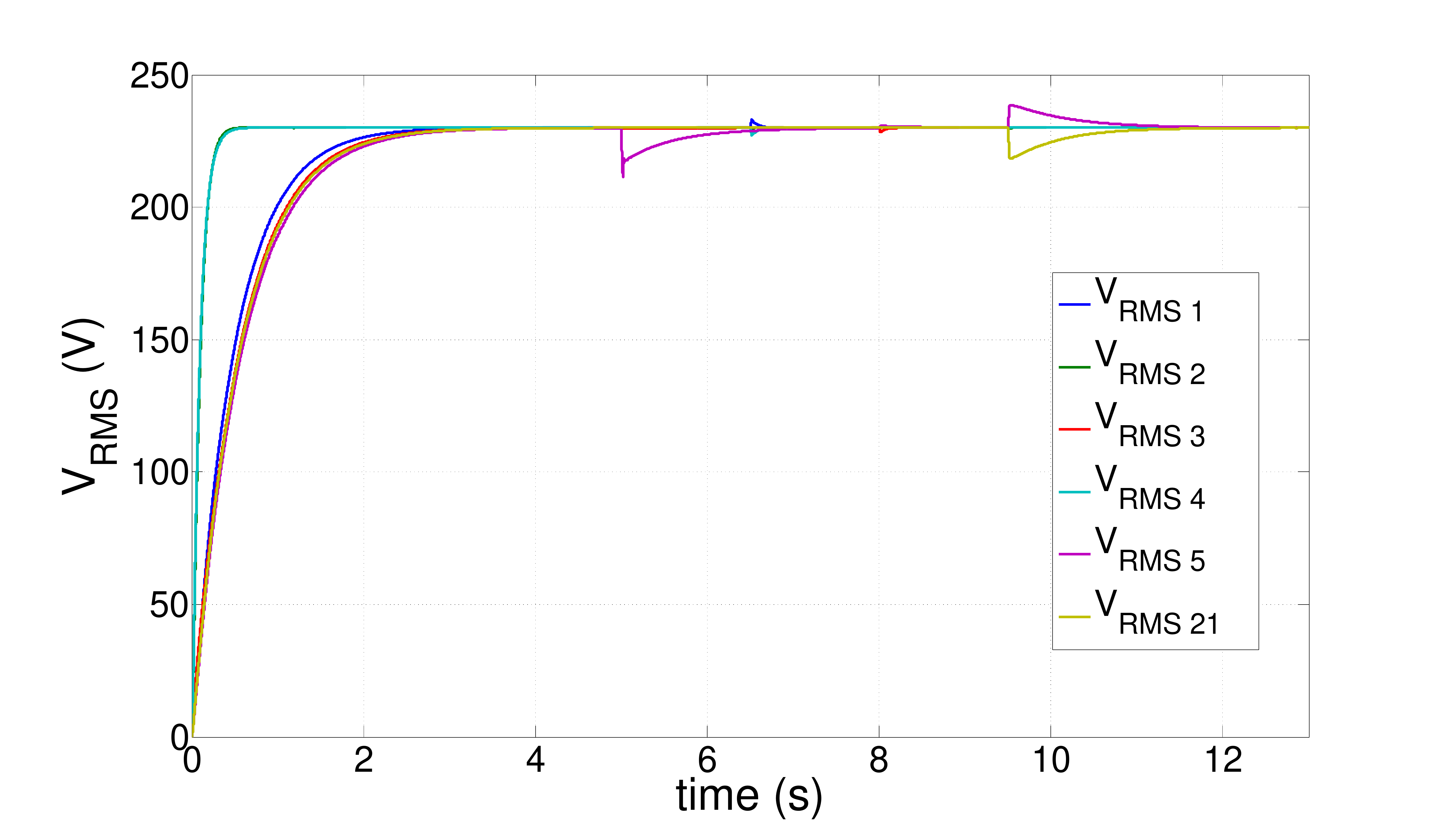}
                        \caption{RMS voltages at boundary nodes.}
                        \label{fig:21nodes_RMS}
                      \end{subfigure}
                      \begin{subfigure}[htb]{0.48\textwidth}
                        \centering
                        \includegraphics[width=1.05\textwidth, height=140pt]{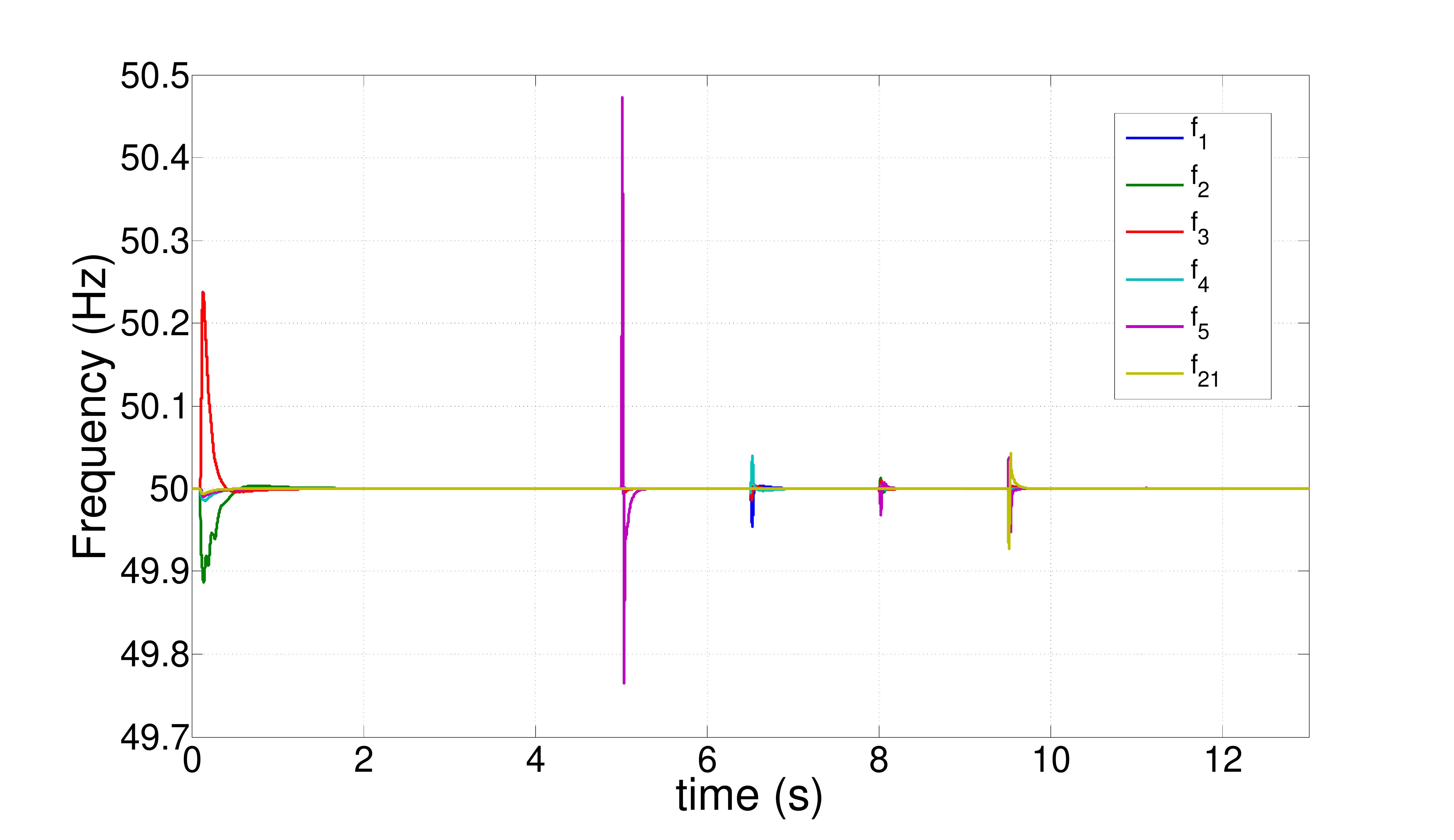}
                        \caption{Frequencies at boundary nodes.}
                        \label{fig:21nodes_f}
                      \end{subfigure}
                      \begin{subfigure}[htb]{0.48\textwidth}
                        \centering
                        \includegraphics[width=1.05\textwidth, height=140pt]{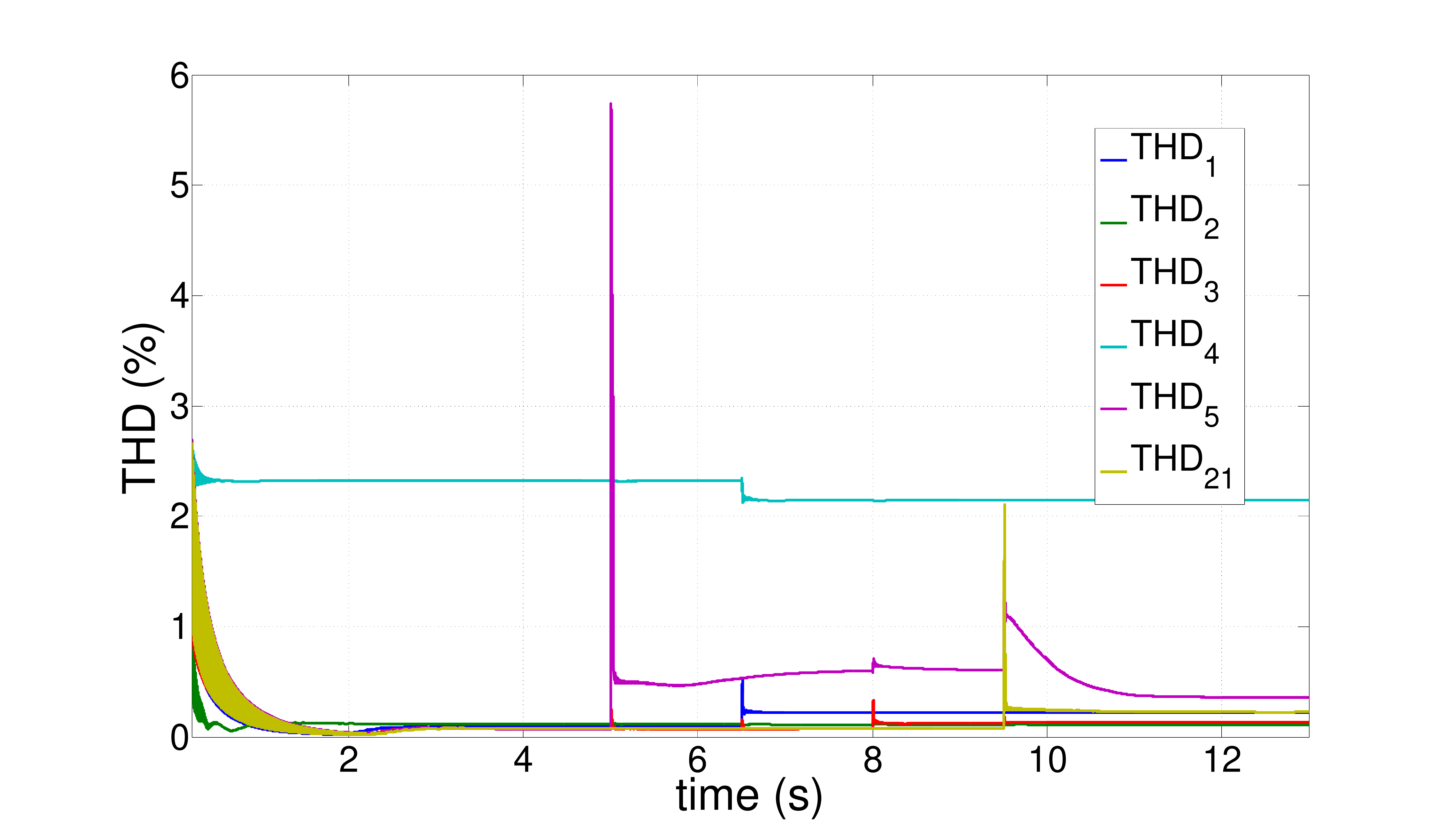}
                        \caption{THD at boundary nodes.}
                        \label{fig:21nodes_THD}
                      \end{subfigure}
                      \begin{subfigure}[htb]{0.48\textwidth}
                        \centering
                        \includegraphics[width=1.05\textwidth, height=140pt]{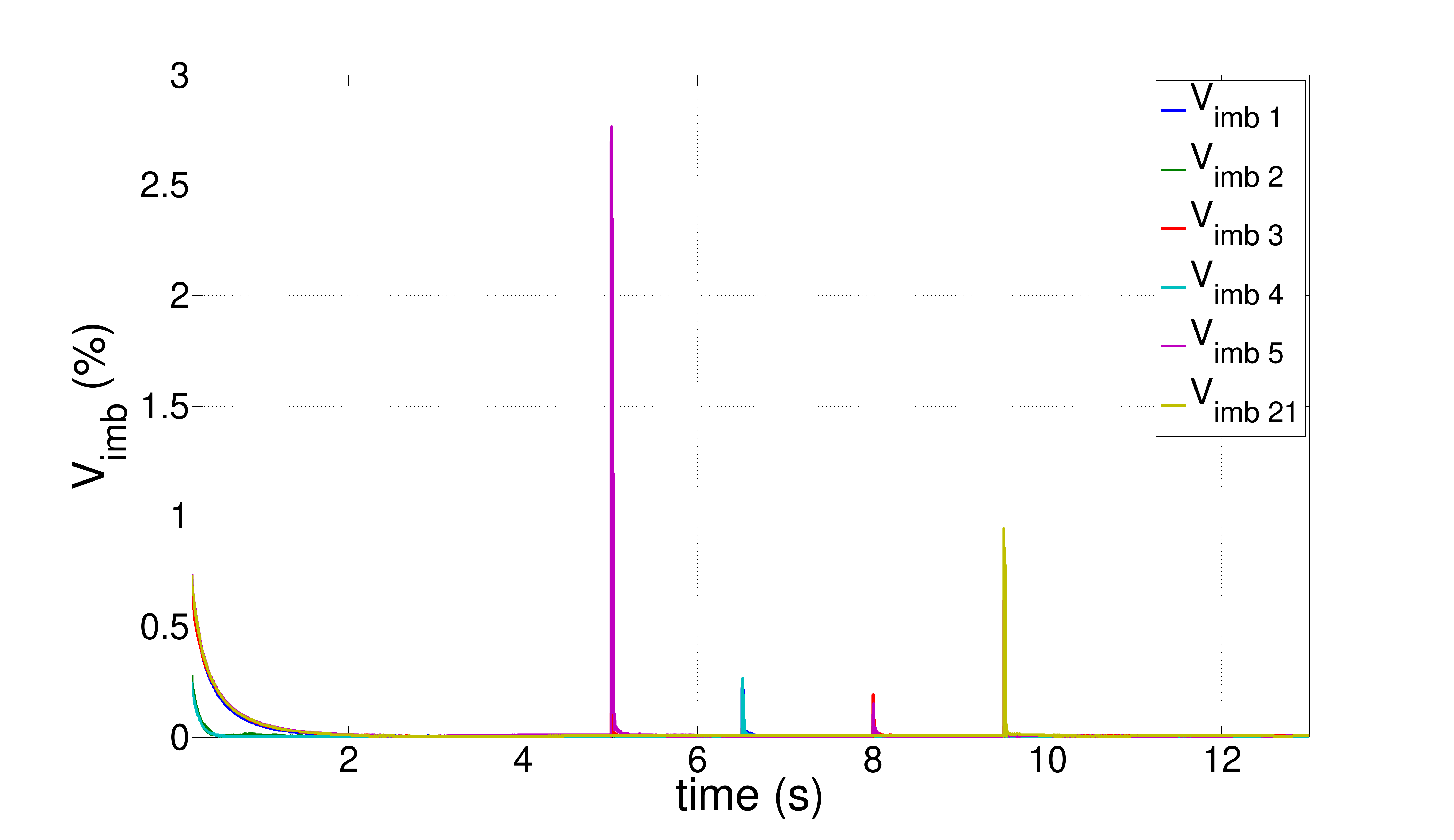}
                        \caption{Voltage imbalance ratio ($V_{N}/V_{P}$).}
                        \label{fig:21nodes_Vimb}
                      \end{subfigure}
                      \begin{subfigure}[htb]{0.48\textwidth}
                        \centering
                        \includegraphics[width=1.05\textwidth, height=140pt]{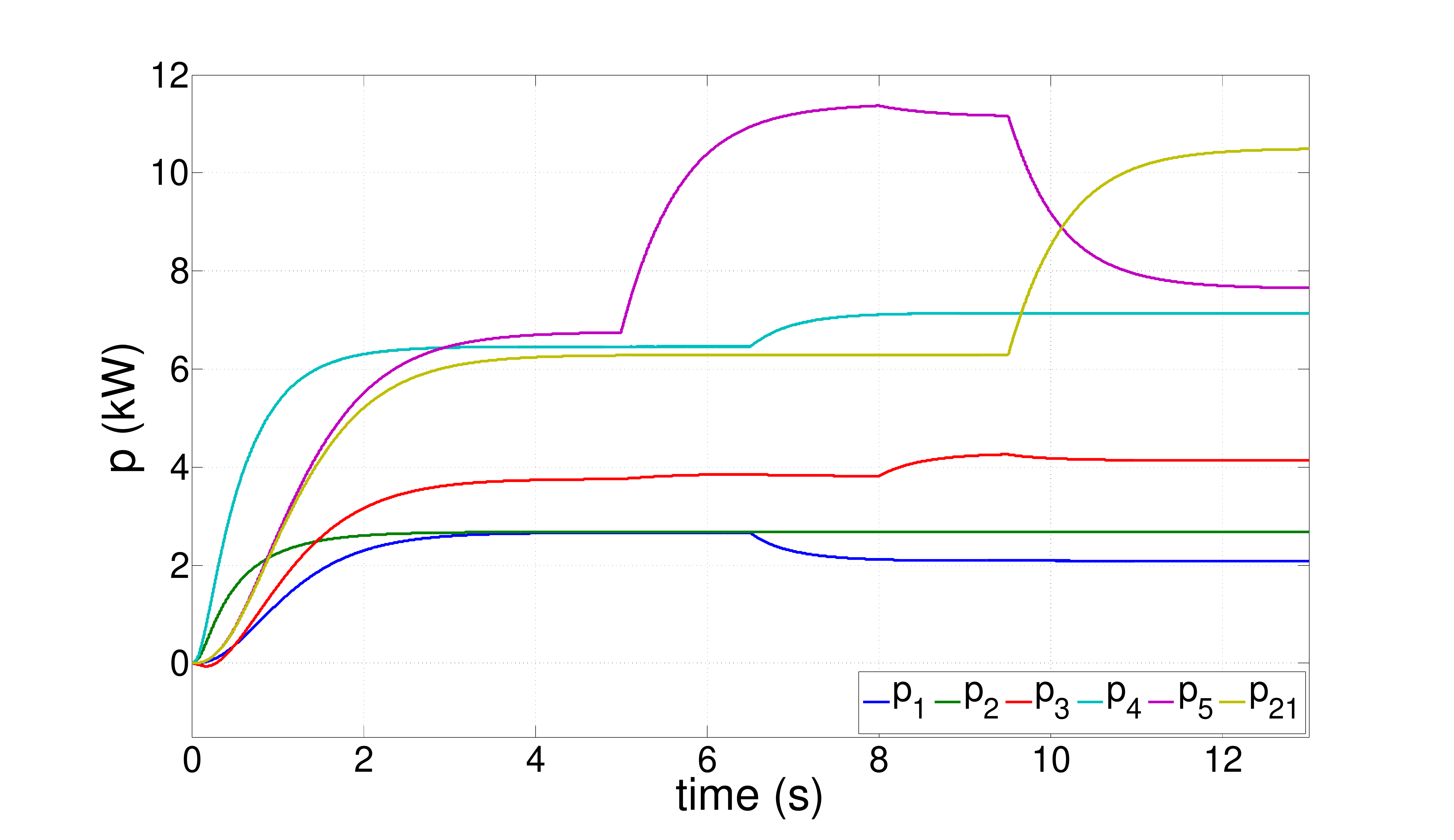}
                        \caption{Active power at PCC nodes.}
                        \label{fig:21nodes_ps}
                      \end{subfigure}
                      \begin{subfigure}[htb]{0.48\textwidth}
                        \centering
                        \includegraphics[width=1.05\textwidth, height=140pt]{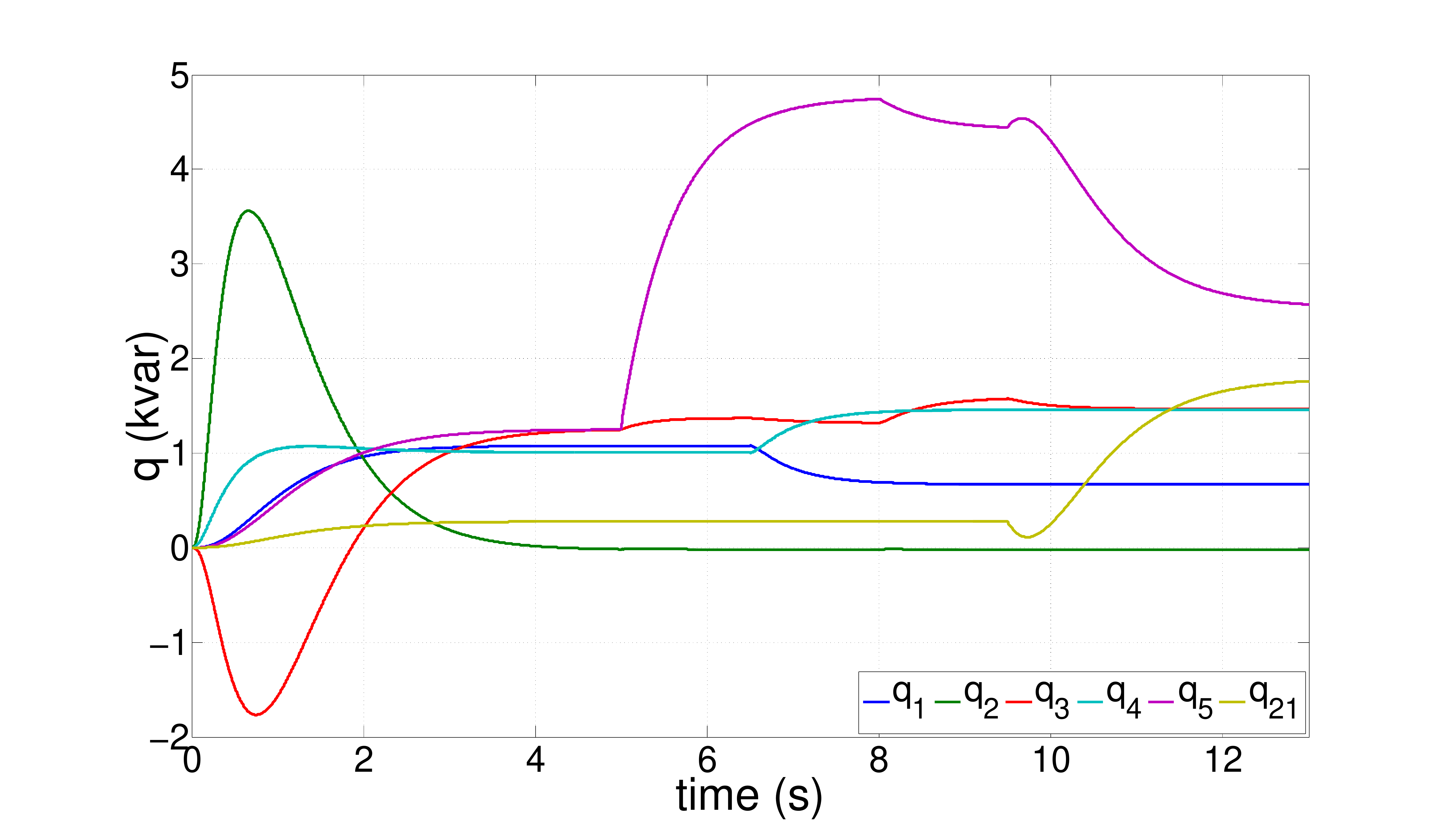}
                        \caption{Reactive power at PCC nodes.}
                        \label{fig:21nodes_qs}
                      \end{subfigure}
                      \caption{Performance of PnP control and AC
                        hybrid KR with a 21-bus network.}
                      \label{fig:21nodes_performance}
                    \end{figure}

   \section{Conclusions}
          \label{sec:conclusions}
          In this paper, we provided a method for extending the PnP control design presented in \cite{riverso2015plug} to AC ImGs with arbitrary topology. We introduced an approximate network reduction algorithm (hKR), based on KR and capable to represent exactly the asymptotic periodic behavior of voltages and currents in the ImG. 

As regards the future developments, besides addressing problems described in Remark \ref{rmk:new}, we aim to extend the hKR also to the case of DC ImGs equipped with PnP controllers \cite{tucci2015plugDC}.

     \clearpage
     \appendix
     \section{Electrical parameters of Examples 1 and 2, Section \ref{subsec:ne}}
\label{app:ne}
		\begin{table}[!h]
			\centering
			\begin{tabular}{ccccc}
					\toprule

					Edge & From node & To node & Resistance [$\Omega$] & Inductance [mH] \\
					\midrule
					$e_1$ & 1 & 4 & 0.1 & 2 \\
					$e_2$ & 2 & 4 & 0.2 & 7 \\
					$e_3$ & 3 & 4 & 1 & 10 \\
					\bottomrule
			\end{tabular}
			\caption{Parameters of the original networks of Examples 1 and 2, Section \ref{subsec:ne}.}
			\label{Appendix:3FTable:OriginalParam}
		\end{table}

	\begin{table}[!h]
			\centering
			\begin{tabular}{ccccc}
				\toprule
				Edge & From node & To node & Resistance [$\Omega$] & Inductance [mH] \\
				\midrule
				$e_{12}$ & 1 & 2 & 0.2746 & 10.354 \\
				$e_{13}$ & 1 & 3 & 1.4482 & 14.8132 \\
				$e_{23}$ & 2 & 3 & 3.9315 & 52.3706 \\
				\bottomrule
			\end{tabular}
			\caption{Parameters of the reduced networks of Examples 1 and 2, Section \ref{subsec:ne}.}
			\label{Appendix:3FTable:ReducedParam}
		\end{table}

	\begin{table}[!h]
			\centering
			\begin{tabular}{cc}
				\toprule
                          Phase & Resistance [$\Omega$] \\
				\midrule
				\textit{a} & 100 \\
				\textit{b} & 10 \\
				\textit{c} & 0.01 \\
				\bottomrule
			\end{tabular}
			\caption{Load parameters of Example 1, Section \ref{subsec:ne}.}
			\label{Appendix:3FTable:Loads}
		\end{table}

	\begin{table}
			\centering
			\begin{tabular}{ccc}
				\toprule
				Name & Description & Values \\ \midrule
				$C_f$ & Filter capacitance of the rectifier & 235 $\mu$F\\
				$L_f$ & Filter inductance of the rectifier & 84 $\mu$H \\
				$R_L$ & Load resistance					& 	40 $\Omega$\\
				\bottomrule
			\end{tabular}
			\caption{Parameters of the six-pulse bridge
                          rectifier of Example 2, Section \ref{subsec:ne}.}
			\label{Appendix:TableNonlinearload}
		\end{table}
		
	\begin{table}
		\centering
		\renewcommand{\arraystretch}{2.0}
		\begin{tabular}{cc}
			\toprule
			Voltage source & Voltage [V] \\ \midrule
			$v_{1a}$ & $\sqrt{\frac{2}{3}}399.8 \sin{\left(\omega_0 t\right)}$ \\
			$v_{2a}$ & $\sqrt{\frac{2}{3}}400 \sin{\left(\omega_0 t\right)}$\\
			$v_{3a}$ & $\sqrt{\frac{2}{3}} 400.2 \sin{\left(\omega_0 t\right)}$\\
			$v_{1b}$ & $\sqrt{\frac{2}{3}}399.8 \sin{\left(\omega_0 t-\frac{2\pi}{3}\right)}$\\
			$v_{2b}$ & $\sqrt{\frac{2}{3}}400 \sin{\left(\omega_0 t-\frac{2\pi}{3}\right)}$\\
			$v_{3b}$ & $\sqrt{\frac{2}{3}} 400.2 \sin{\left(\omega_0 t-\frac{2\pi}{3}\right)}$\\
			$v_{1c}$ & $\sqrt{\frac{2}{3}}399.8 \sin{\left(\omega_0 t+\frac{2\pi}{3}\right)}$\\
			$v_{2c}$ & $\sqrt{\frac{2}{3}} 400 \sin{\left(\omega_0 t+\frac{2\pi}{3}\right)}$\\
			$v_{3c}$ & $\sqrt{\frac{2}{3}} 400.2 \sin{\left(\omega_0 t+\frac{2\pi}{3}\right)}$\\
			\bottomrule
		\end{tabular}
		\caption{Voltage values of generators in Examples 1 and 2, Section \ref{subsec:ne}.}
		\label{Appendix:TableVoltages}
	\end{table}
     \clearpage
     \section{21-bus network}
		\label{Appendix:label}
		This appendix lists the electrical parameters of the
                21-bus network shown in Section \ref{sec:Sim21bn}. 
		The resistances and inductances of the RL lines of network in Figure \ref{Fig6:Retemia} are collected in Table \ref{Appendix:Table21original}.
		\begin{table}[!h]
			\centering
			\begin{tabular}{ccccc}
				\toprule
				Edge & From node & To node & Resistance [$\Omega$] & Inductance [mH] \\
				\midrule
				$e_1$ & 1 & 7 & 1 & 20 \\
				$e_2$ & 6 & 7 & 0.1 & 1.8 \\
				$e_3$ & 3 & 7 & 1 & 200 \\
				$e_4$ & 2 & 8 & 0.6 & 6 \\
				$e_5$ & 8 & 3 & 0.4 & 35 \\
				$e_6$ & 3 & 12 & 0.1 & 1.8 \\
				$e_7$ & 3 & 4 & 1 & 600 \\
				$e_8$ & 4 & 11 & 0.1 & 2 \\
				$e_9$ & 10 & 11 & 0.1 & 2.5\\
				$e_{10}$ & 9 & 10 & 0.2 & 4.5 \\
				$e_{11}$ & 12 & 13 & 1.1 & 300 \\
				$e_{12}$ & 13 & 15 & 1 & 40 \\
				$e_{13}$ & 13 & 14 & 0.1 & 2 \\
				$e_{14}$ & 14 & 5 & 0.3 & 8 \\
				$e_{15}$ & 14 & 18 & 0.1 & 1 \\
				$e_{16}$ & 16 & 18 & 0.3 & 30 \\
				$e_{17}$ & 17 & 18 & 0.1 & 2 \\
				$e_{18}$ & 6 & 10 & 1.1 & 20 \\
				$e_{19}$ & 12 & 16 & 2.1 & 300 \\
				$e_{20}$ & 14 & 19 & 0.5 & 10 \\
				$e_{21}$ & 19 & 20 & 0.3 & 7 \\
				$e_{22}$ & 19 & 21 & 0.1 & 1.8  \\
				\bottomrule
			\end{tabular}
			\caption{Parameters of the original 21-bus network.}
			\label{Appendix:Table21original}
		\end{table}
		The loads connected to the buses are listed in Tables \ref{Appendix:Table21Linloads} and \ref{Appendix:Table21loads}.
		\begin{table}[!h]
			\centering
			\begin{tabular}{ccc}
				\toprule
				Node & Resistance [$\Omega$] & Inductance [H] \\
				\midrule
				6 & 80 & 0.2 \\
				7 & 80 & 0 \\
				8 & 50 & 0 \\
				10 & 100 & 0 \\
				11 & 100 & 0 \\
				12 & 50 & 0.05 \\
				13 & 100 & 0 \\
				14 & 50 & 0 \\
				15 & 60 & 0 \\
				17 & 3 & 1 \\
				18 & 45 & 0.02 \\
				19 & 50 & 0 \\
				20 & 50 & 0\\
				\bottomrule
			\end{tabular}
			\caption{Linear loads parameters.}
			\label{Appendix:Table21Linloads}
			\vspace{1cm}
			
			\begin{tabular}{cccc}
				\toprule
				Node & Resistance [$\Omega$] & Filter Inductance [$\mu$H] & Filter Capacitance [$\mu$F]\\
				\midrule
				9 & 80 & 84 & 235\\
				16 & 80 & 84 & 235\\
				\bottomrule
			\end{tabular}
			\caption{Nonlinear loads connected to the buses.}
			\label{Appendix:Table21loads}
		\end{table}
		
	As regards the parameters of the hybrid Kron reduced networks
        shown in Figure \ref{fig:21nodes_KR}, we highlight that they
        are all positive (i.e. the nature of the original circuits is
        preserved). In particular, the reduced parameters, when all the four switches are open, are listed in Table \ref{Appendix:Table21EqParam1}. Table \ref{Appendix:Table21EqParam2} shows the reduced R and L when $SW_1$ and $SW_2$ are closed, with $SW_3$ and $SW_4$ open. The reduced parameters, relative to the network with $SW_1$, $SW_2$ and $SW_3$ closed and $SW_4$ open, are collected in Table \ref{Appendix:Table21EqParam3}. Finally, Table \ref{Appendix:Table21EqParam4} is referred to the case with all four switches closed.
		\begin{table}[!h]
			\centering
			\begin{tabular}{ccccc}
				\toprule
				Edge & From node & To node & Resistance [$\Omega$] & Inductance [mH] \\
				\midrule
				$r_1$ & 1 & 3 & 2 & 220\\
				$r_2$ & 2 & 3 & 1 & 41\\
				$r_3$ & 3 & 4 & 1 &  600 \\
				$r_4$ & 3 & 5 & 1.6 & 311.8 \\
				\bottomrule
			\end{tabular}
			\caption{Equivalent parameters when $SW_2$, $SW_3$ and $SW_4$ are open.}
			\label{Appendix:Table21EqParam1}
		\end{table}
		
		\begin{table}[!h]
			\centering
			\begin{tabular}{ccccc}
				\toprule
				Edge & From node & To node & Resistance [$\Omega$] & Inductance [mH] \\
				\midrule
				$r_1$ & 1 & 3 & 2.2813 & 371.9\\
				$r_5$ & 1 & 4 & 2.6589 & 48.9\\
				$r_2$ & 2 & 3 & 1 & 41\\
				$r_3$ & 3 & 4 & 1.58 &  269.7 \\
				$r_4$ & 3 & 5 & 1.6 & 311.8 \\
				\bottomrule
			\end{tabular}
			\caption{Equivalent parameters when $SW_1$ $SW_2$ are closed, while $SW_3$ and $SW_4$ are open.}
			\label{Appendix:Table21EqParam2}
		\end{table}
		
		\begin{table}[!h]
			\centering
			\begin{tabular}{ccccc}
				\toprule
				Edge & From node & To node & Resistance [$\Omega$] & Inductance [mH] \\
				\midrule
				$r_1$ & 1 & 3 & 2.2813 & 371.9\\
				$r_5$ & 1 & 4 & 2.6586  &   48.9 \\
				$r_2$ & 2 & 3 & 1 & 41\\
				$r_3$ & 3 & 4 & 1.58 &  269.7 \\
				$r_4$ & 3 & 5 & 1.2971 & 167.7 \\
				\bottomrule
			\end{tabular}
			\caption{Equivalent parameters when $SW_1$, $SW_2$ and $SW_3$ are closed, and $SW_4$ is open.}
			\label{Appendix:Table21EqParam3}
		\end{table}
		\begin{table}[!h]
			\centering
			\begin{tabular}{ccccc}
				\toprule
				Edge & From node & To node & Resistance [$\Omega$] & Inductance [mH] \\
				\midrule
				$r_1$ & 1 & 3 & 2.2813 & 371.9\\
				$r_5$ & 1 & 4 &  2.6586 &  48.9  \\
				$r_2$ & 2 & 3 & 1 & 41\\
				$r_3$ & 3 & 4 & 1.58 &  269.7 \\
				$r_4$ & 3 & 5 & 0.5602 & 275.4 \\
				$r_6$ & 3 & 21 &  6.1756 & 408.1 \\
				$r_7$ & 5 & 21 & 0.9486  &   20.4\\
				\bottomrule
			\end{tabular}
			\caption{Equivalent parameters when all the four switches are open.}
			\label{Appendix:Table21EqParam4}
		\end{table}

     \clearpage
 \bibliographystyle{IEEEtran}
     \bibliography{PnP_AC_ImG_Kron}

\end{document}